\documentclass[11pt,numbers,sort&compress]{elsarticle}
\journal{}

\makeatletter
\def\ps@pprintTitle{%
 \let\@oddhead\@empty
 \let\@evenhead\@empty
 \def\@oddfoot{\hfill\thepage}%
 \let\@evenfoot\@oddfoot}
\makeatother

\usepackage{amsmath,amssymb,amsfonts,amsthm,graphicx,hyperref}
\usepackage[labelfont=bf]{caption}
\providecommand{\doi}[1]{\href{https://doi.org/#1}{DOI:#1}}
\usepackage{xurl} 
\renewcommand{\doi}[1]{%
 \href{https://doi.org/#1}{\nolinkurl{DOI:#1}}%
}

\usepackage{enumitem} 
\usepackage{appendix} 
\usepackage{natbib} 
\usepackage{dsfont} 
\usepackage{mathtools} 
\usepackage{xcolor} 
\usepackage{multirow} 
\usepackage{booktabs} 
\usepackage{float} 
\usepackage[labelfont=bf]{caption}
\usepackage{subcaption} 
\usepackage{geometry} 
\newtheorem{theorem}{Theorem}[section]
\newtheorem{proposition}[theorem]{Proposition}
\newtheorem{lemma}[theorem]{Lemma}
\newtheorem{corollary}[theorem]{Corollary}
\newtheorem{remark}{Remark}

\numberwithin{equation}{section}

\newcommand{\N}{\mathbb{N}}
\newcommand{\R}{\mathbb{R}}
\newcommand{\PP}{\mathsf{P}} 
\newcommand{\EE}{\mathsf{E}} 
\newcommand{\Bias}{\mathsf{Bias}} 
\newcommand{\Var}{\mathsf{Var}} 
\newcommand{\Cov}{\mathsf{Cov}} 
\newcommand{\bb}[1]{\boldsymbol{#1}}
\newcommand{\OO}{\mathcal{O}}
\newcommand{\oo}{\mathrm{o}}
\newcommand{\rd}{\mathrm{d}}
\newcommand{\ind}{\mathds{1}}
\newcommand{\e}{\varepsilon}

\allowdisplaybreaks

\begin{document}

\begin{frontmatter}
\title{Dirichlet kernel density estimation on the simplex with missing data}

\author[a1]{Sami Baraket}\ead{SMBaraket@imamu.edu.sa}
\author[a2]{Hanen Daayeb}\ead{hanen.daayeb@fst.utm.tn}
\author[a2]{Salah Khardani}\ead{salah.khardani@fst.utm.tn}
\author[a3]{Guanjie Lyu}\ead{glyu@dal.ca}
\author[a4]{Fr\'ed\'eric Ouimet}\ead{frederic.ouimet2@uqtr.ca}

\address[a1]{Department of Mathematics and Statistics, College of  Science, Imam Mohammad Ibn Saud Islamic University, Riyadh, 11623, Saudi Arabia}
\address[a2]{Laboratoire des \'equations aux d\'eriv\'ees partielles, Universit\'e de Tunis El-Manar, 2092 El Manar 2, Tunis, Tunisia}
\address[a3]{Department of Community Health and Epidemiology, Dalhousie University, Canada}
\address[a4]{D\'epartement de math\'ematiques et d'informatique, Universit\'e du Qu\'ebec \`a Trois-Rivi\`eres, Canada\vspace{-5mm}}

\begin{abstract}
Nonparametric density estimation for compositional data supported on the simplex is examined under a missing at random mechanism. Rather than imputing missing values and estimating the density from a completed data set, we adopt a strategy based on inverse probability weighting. The proposed estimator uses an adaptive Dirichlet kernel, which ensures nonnegativity on the simplex and favorable behavior near the boundary. When the observation probabilities are unknown, they are estimated through a Nadaraya--Watson regression step. The large-sample properties of the estimator are derived, including pointwise bias and variance expansions, optimal smoothing rates, and asymptotic normality. A simulation study investigates its finite-sample performance under varying sample sizes and missing rates. Simulations show our method outperforms inverse-probability-weighted kernel density estimators based on additive and isometric log-ratio transformations of the data for certain target densities. The methodology is further illustrated through an application to leukocyte composition data from the National Health and Nutrition Examination Survey (NHANES), which allows for the identification of the modal immune profile in the sampled population.
\end{abstract}


\end{frontmatter}

\section{Introduction}\label{sec:introduction}

Compositional data, consisting of non-negative components that sum to one, arise in a broad range of scientific fields. Typical examples include geochemical compositions in Earth and environmental sciences~\citep{Carranza2011LogratioStreamSediment}, microbiome relative-abundance profiles in biomedical research~\citep{HMP2012Nature}, dietary intake~\citep{Leite2016NutritionalCoDA}, time-use shares in public health \citep{Chastin2015CoDA}, and portfolio allocation vectors in finance~\citep{VegaGamezAlonsoGonzalez2024horizons}. Because such data encode parts of a whole, their support lies on the probability simplex, and standard multivariate techniques often fail to account for the inherent closure constraint and the induced dependence among components. The foundational work of \citet{Aitchison1982statistical,Aitchison1986Book} established the modern statistical treatment of compositional data and motivated a substantial literature across the natural and health sciences. Recent applications in microbiome studies~\citep{Gloor2017microbiome} and geochemistry~\citep{Grunsky2019geochemistry} illustrate the growing relevance of compositional modeling in contemporary data analysis.

From a statistical perspective, the simplex geometry has motivated the development of specialized models and inferential tools. A substantial line of work focuses on log-ratio transformations \citep{Aitchison1986Book,Egozcue2003isometric}, which map the simplex to Euclidean space and thereby facilitate regression and multivariate modeling. Alternative approaches model compositions directly on the simplex, for example through Dirichlet regression for covariate effects and Dirichlet-mixture models for flexible multi-modal distributions~\citep{Gueorguieva2008dirichlet,Pal2022Bayesian}. For density estimation on the simplex, recent contributions include Bernstein polynomial estimators \citep{MR1293514,MR4287788,doi:10.1515/stat-2022-0111,MR4796622} and Dirichlet kernel methods \citep{doi:10.2307/2347365,doi:10.1016/j.cageo.2009.12.011,MR4319409,MR4544604,arXiv:2510.07608,DKO2026}. These works provide a nonparametric framework for density estimation on the simplex, establish minimax and boundary properties, and extend to settings with strong mixing dependence. Related Dirichlet kernel regression estimators for simplex-valued responses, including Nadaraya--Watson and local polynomial variants, have also been developed~\citep{MR4796622,MR4905615,arXiv:2502.08461}. For a unified treatment of asymmetric kernel density estimation methods, see, e.g., \cite{AboubacarKokonendji2025recursive,KokonendjiSome2018associated,KokonendjiSome2021bayesian,EsstafaKokonendjiNgo2026associated}.

In empirical research, compositional data frequently exhibit missingness mechanisms in which the probability of an observation being unobserved depends on fully observed covariates, a setting known as missing at random (MAR) as formalized by \citet{Rubin1976inference}. In microbiome studies, for example, missing relative abundance profiles are often systematically linked to observed technical factors such as sequencing depth or to clinical metadata, rather than to the unobserved abundance of the taxa themselves~\citep{Jiang2021mbimpute,Peleg2024interpolation}. Similar MAR mechanisms arise in geochemistry and public health applications (e.g., time-use or nutritional surveys), where assay failures or survey nonresponse may be predicted by observed demographic variables or sample-specific characteristics~\citep{MartinFernandez2003zeros,PalareaAlbaladejo2013detection,Abraham2006ATUS,GrovesPeytcheva2008nonresponse,Fakhouri2020nhanes}. In such settings, the simplifying assumption that missingness occurs completely at random, meaning that the probability of response is unrelated to the unobserved data, is rarely plausible. Relying solely on complete observations (complete-case analysis) therefore risks introducing bias and discarding useful information.

A common approach for handling missing responses is imputation, which generates completed data sets that can then be analyzed using standard procedures~\citep{Hron2010imputation,Little2019Book}. Applied to density estimation, this strategy reconstructs the missing values first and only then estimates the density from the completed sample. While flexible, this two-stage approach requires modeling how the missing responses relate to the observed data and therefore operates indirectly on the distributional target of interest. In contrast, inverse probability weighting addresses missingness by reweighting the observed responses according to their probabilities of being observed. Under the MAR assumption, this idea, originating from the Horvitz--Thompson estimator in survey sampling~\citep{Horvitz1952generalization} and subsequently developed in the nonparametric and semiparametric missing-data literature~\citep{Robins1994correcting,Tsiatis2006semiparametric,Dubnicka2009kernel,Gharbi2025Bernstein,SeamanWhite2013IPW}, yields a direct reconstruction of the full-data distribution without filling in individual values. This makes inverse probability weighting a natural strategy for nonparametric density estimation on the simplex, where the underlying distribution can be complex and the support is constrained.

Building on this perspective, a Dirichlet kernel–based density estimator is introduced for simplex-valued responses subject to MAR mechanisms. The estimator employs inverse probability weighting to correct for selection bias and incorporates a Nadaraya--Watson smoothing step to estimate the observation probabilities (also called propensity scores) when they are unknown, while respecting the simplex geometry and mitigating boundary effects inherent in classical kernel smoothing methods. The resulting procedure admits a full asymptotic characterization, including pointwise bias and variance expansions, optimal smoothing rates, and asymptotic normality for both the pseudo estimator (when the propensities are known) and its feasible counterpart (when the propensities are estimated). Taken together, these results extend Dirichlet kernel density estimation to incomplete-data settings and clarify the efficiency and robustness implications of applying inverse probability weighting to nonparametric density estimation on the simplex.

The remainder of the paper is organized as follows. Section~\ref{sec:definitions.notation} introduces the simplex framework, the MAR missingness setup, and the inverse probability weighted (IPW) Dirichlet kernel estimators (pseudo and feasible). Section~\ref{sec:assumptions} states the regularity conditions on the target density, covariates, propensity score, and smoothing parameters used for the asymptotic analysis. Section~\ref{sec:main.results} presents the main theoretical results, including pointwise bias and variance expansions, mean squared error (MSE) rates, and asymptotic normality for both the pseudo estimator $\widetilde{f}_{n,b}$ and the feasible estimator $\hat{f}_{n,b}$. Section~\ref{sec:simulation} reports simulation experiments (including bandwidth selection) and compares the proposed IPW Dirichlet kernel density estimator (KDE) with log-ratio-based IPW alternatives under varying sample sizes and missing rates. Section~\ref{sec:application} illustrates the methodology on NHANES leukocyte composition data and highlights the estimated modal composition. Section~\ref{sec:outlook} summarizes the contributions and discusses directions for extensions and further research. Section~\ref{sec:proofs} contains the proofs of the main propositions and theorems, while Section~\ref{sec:tech.lemmas} gathers auxiliary technical lemmas used in the proofs. \ref{app:acronyms} provides a list of acronyms used throughout. To ensure reproducibility of our findings, \ref{app:code} links to the \textsf{R} code used to generate the figures, the simulation results, and the real-data application.

\section{Definitions and notation}\label{sec:definitions.notation}

For any integer $d\in \N = \{1,2,\ldots\}$, the $d$-dimensional simplex and its interior are defined by
\[
\mathcal{S}_d = \{\bb{s}\in [0,1]^d : \|\bb{s}\|_1 \leq 1\}, \qquad \mathrm{Int}(\mathcal{S}_d) = \{\bb{s}\in (0,1)^d : \|\bb{s}\|_1 < 1\},
\]
where $\|\bb{s}\|_1 = \sum_{i=1}^d |s_i|$ denotes the $\ell^1$ norm in $\R^d$. For any $\alpha_1,\ldots,\alpha_d,\beta\in (0,\infty)$, the density of the $\mathrm{Dirichlet}\hspace{0.2mm}(\bb{\alpha},\beta)$ distribution is
\[
K_{\bb{\alpha},\beta}(\bb{s}) = \frac{\Gamma(\|\bb{\alpha}\|_1 + \beta)}{\Gamma(\beta) \prod_{i=1}^d \Gamma(\alpha_i)} (1 - \|\bb{s}\|_1)^{\beta - 1} \prod_{i=1}^d s_i^{\alpha_i - 1}, \quad \bb{s}\in \mathcal{S}_d.
\]

Let $\bb{Y}\in \mathcal{S}_d$ be a continuous response with unknown density $f$, and let $\bb{X} = (X_1,\ldots,X_p)\in \R^p$ be a fully observed continuous covariate vector with density $g$. Assume that the pair $(\bb{X},\bb{Y})$ admits a joint density $h$ supported on $\R^p\times \mathcal{S}_d$. The available data consist of $n$ independent and identically distributed (iid) copies $(\bb{X}_1,\bb{Y}_1),\ldots,(\bb{X}_n,\bb{Y}_n)$ of $(\bb{X},\bb{Y})$, where the responses $\bb{Y}_i$ are assumed to be MAR.

To track which responses are observed, define the indicator
\[
\delta_i = \ind_{\{\bb{Y}_i \text{ is observed}\}} =
\begin{cases}
1, & \text{if } \bb{Y}_i \text{ is observed},\\
0, & \text{otherwise},
\end{cases}
\]
and assume that its conditional expectation, called the {\it propensity score}, satisfies
\[
\pi(\bb{X}_i) = \PP\left(\delta_i = 1 \mid \bb{Y}_i, \bb{X}_i\right) = \PP\left(\delta_i=1\mid \bb{X}_i\right).
\]
In particular, $\delta_i$ is conditionally independent of $\bb{Y}_i$ given $\bb{X}_i$.

If all the sampled responses $\bb{Y}_1,\ldots,\bb{Y}_n$ are observed, then for an appropriately chosen smoothing parameter $b\in (0,\infty)$, one possible estimator of the unknown density $f$ is the full-data Dirichlet KDE proposed by \citet{doi:10.2307/2347365} and studied by \cite{MR4319409,MR4544604,DKO2026}, defined by
\[
\hat{f}^{\hspace{0.2mm}\mathrm{full}}_{n,b}(\bb{s}) = n^{-1} \sum_{i=1}^n \kappa_{\bb{s},b}(\bb{Y}_i), \quad \bb{s}\in \mathcal{S}_d,
\]
where, for brevity,
\[
\kappa_{\bb{s},b}(\cdot) = K_{\bb{s}/b + \bb{1}, (1 - \|\bb{s}\|_1)/b + 1}(\cdot),
\]
with $\bb{1} = (1,\ldots,1)$ being the $d$-vector of ones. If some responses $\bb{Y}_i$ are missing, a straightforward fix is to estimate $f$ using only the data at hand, in which case $\hat{f}^{\hspace{0.2mm}\mathrm{full}}_{n,b}$ reduces to the {\it complete-case estimator}:
\[
\hat{f}^{\hspace{0.2mm}\mathrm{cc}}_{n,b}(\bb{s}) = \frac{1}{\sum_{j=1}^n \delta_j} \sum_{i=1}^n \delta_i \, \kappa_{\bb{s},b}(\bb{Y}_i), \quad \bb{s}\in \mathcal{S}_d.
\]
However, this estimator is generally biased under missingness mechanisms beyond missing completely at random. In particular, it is biased under the MAR mechanism we are investigating in this paper. To counter this effect, an IPW pseudo estimator can be used:
\begin{equation}\label{eq:pseudo}
\widetilde{f}_{n,b}(\bb{s}) = n^{-1} \sum_{i=1}^n \frac{\delta_{i}}{\pi(\bb{X}_i)}\kappa_{\bb{s},b}(\bb{Y}_i), \quad \bb{s}\in \mathcal{S}_d.
\end{equation}

In most practical situations, however, the propensity scores $\pi(\bb{X}_i)$ are unknown, unless they are specified by design, as in planned missingness schemes. Consequently, $\smash{\widetilde{f}_{n,b}}$ does not constitute a genuine estimator of $f$. This motivates using a Nadaraya--Watson regression estimator of $\pi(\bb{X}_i)$, defined by
\begin{equation}\label{eq:NW}
\hat{\pi}_i(\bb{X}_{1:n}) = \frac{\sum_{j=1}^n \delta_j \, K_h^*(\bb{X}_i - \bb{X}_j)} {\sum_{j=1}^n K_h^*(\bb{X}_i - \bb{X}_j)}, \qquad K_h^*(\bb{x}) = h^{-p} K^*(h^{-1} \bb{x}),
\end{equation}
where $K^*: \R^p \to [0,\infty)$ denotes a classical kernel function and $h$ the associated smoothing parameter. By replacing $\pi(\bb{X}_i)$ with $\smash{\hat{\pi}_i(\bb{X}_{1:n})}$ in \eqref{eq:pseudo}, we obtain the following {\it feasible estimator}:
\begin{equation}\label{eq:feasible}
 \hat{f}_{n,b}(\bb{s})= n^{-1} \sum_{i=1}^n \frac{\delta_{i}}{\hat{\pi}_i(\bb{X}_{1:n})}\kappa_{\bb{s},b}(\bb{Y}_i), \quad \bb{s}\in \mathcal{S}_d.
\end{equation}

Throughout the paper, the following notational conventions are adopted. The notation $u = \OO(v)$ means that $\limsup |u / v| \leq C < \infty$ as $b\to 0$ or $n\to \infty$, depending on the context. The positive constant $C$ may depend on the target density function $f$ and the dimensions $d$ and $p$, but on no other variables unless explicitly written as a subscript. A common occurrence is a local dependence of the asymptotics on a given point~$\bb{s}$ on the simplex, in which case one writes $u = \OO_{\bb{s}}(v)$. The alternative notation $u \ll v$ is also used to mean $u,v\geq 0$ and $u = \OO(v)$. If both $u \ll v$ and $u \gg v$ hold, one writes $u \asymp v$. Similarly, the notation $u = \oo(v)$ means that $\lim |u / v| = 0$ as $b\to 0$ or $n\to \infty$. Subscripts indicate which parameters the convergence rate can depend on. The symbol $\rightsquigarrow$ denotes convergence in law. The shorthand notations $\bb{X}_{1:n} = (\bb{X}_i)_{i=1}^n$ and $[k] = \{1,\ldots,k\}$ for any $k\in\N$ will be used frequently. The smoothing parameter $b = b(n)$ is always implicitly a function of the number of observations, except in Section~\ref{sec:tech.lemmas}.

\section{Assumptions}\label{sec:assumptions}

In addition to the baseline assumptions introduced in Section~\ref{sec:definitions.notation}, the assumptions below are imposed to derive the results stated in Section~\ref{sec:main.results}.

\begin{enumerate}[label=(A\arabic*)]\setlength\itemsep{0em}
\item $b = b(n)\to 0$ and $n b^{d/2} \to \infty$ as $n\to \infty$. \label{ass:1}
\item The density $f$ of $\bb{Y}$ is Lipschitz continuous on $\mathcal{S}_d$. \label{ass:2}
\item The density $f$ of $\bb{Y}$ is twice continuously differentiable on $\mathcal{S}_d$. \label{ass:3}
\item The marginal density $g$ of $\bb{X}$ has bounded support and $\inf_{\bb{x}:g(\bb{x})>0} g(\bb{x}) \equiv g_{\mathrm{min}}\in (0,\infty)$. \label{ass:4}
\item The propensity function $\pi$ is bounded away from zero, i.e., $\min_{\bb{x}:g(\bb{x}) > 0} \pi(\bb{x}) \equiv \pi_{\min}\in (0,\infty)$. \label{ass:5}
\item For all $\bb{x}\in \R^p$ such that $g(\bb{x})\in (0,\infty)$, the conditional density $f_{\bb{Y}_1 \mid \bb{X}_1}(\cdot \mid \bb{x})$ is Lipschitz continuous on $\mathcal{S}_d$ with Lipschitz constant $L(\bb{x})$, and $\sup_{\bb{x}: g(\bb{x}) > 0} L(\bb{x}) < \infty$. We also assume that $f_{\bb{Y}_1 \mid \bb{X}_1}(\cdot \mid \bb{x})$ is bounded uniformly over $\bb{x}\in \R^p$ such that $g(\bb{x})\in (0,\infty)$. \label{ass:6}
\end{enumerate}

\par\noindent\rule{\textwidth}{0.5pt}

\begin{enumerate}[label=(B\arabic*)]\setlength\itemsep{0em}
\item $h = h(n)\to 0$ and $n h^p \to \infty$ as $n\to \infty$. \label{ass:B.1}
\item The propensity score function $\pi$ and the density $g$ of $\bb{X}$ are bounded and all their first- and second-order partial derivatives are continuous and bounded on the support of $\bb{X}$. \label{ass:B.2}
\item The function $K^*:\R^p \to [0,\infty)$ is bounded, symmetric, and satisfies, for all $i,j\in [p] ~(i\neq j)$:
\[
\int_{\R^p} K^*(\bb{u}) \, \rd \bb{u} = 1, \quad
\int_{\R^p} \bb{u} \, K^*(\bb{u})\, \rd \bb{u} = \bb{0}_p, \quad
\int_{\R^p} u_i u_j K^*(\bb{u}) \, \rd \bb{u} = 0, \quad
\int_{\R^p} u_i^2 K^*(\bb{u}) \, \rd \bb{u} < \infty.
\] \label{ass:B.3}
\end{enumerate}

\vspace{-6mm}
\par\noindent\rule{\textwidth}{0.5pt}

\begin{enumerate}[label=(C\arabic*)]\setlength\itemsep{0em}
\item The function $q(\bb{x}) = \EE[\kappa_{\bb{s},b}(\bb{Y}_1) \mid \bb{X}_1 = \bb{x}]$ is bounded and all its first- and second-order partial derivatives are continuous and bounded on the support of $\bb{X}$. \label{ass:C.1}
\end{enumerate}

\section{Main results}\label{sec:main.results}

\subsection{Results for the pseudo estimator \texorpdfstring{$\widetilde{f}_{n,b}$}{widetilde(f)\_n,b}}\label{sec:results.pseudo}

\begin{proposition}[Pointwise bias]\label{prop:bias.pseudo}
Suppose that Assumptions~\ref{ass:1}, \ref{ass:3}, and \ref{ass:5} hold. Uniformly for $\bb{s}\in \mathcal{S}_d$, we have
\[
\Bias[\widetilde{f}_{n,b}(\bb{s})] = \Bias[\hat{f}^{\hspace{0.2mm}\mathrm{full}}_{n,b}(\bb{s})] = b \, \phi(\bb{s}) + \oo(b), \quad n\to \infty,
\]
where
\begin{equation}\label{eq:phi}
\phi(\bb{s}) = \sum_{i\in [d]} (1 - (d+1)s_i) \frac{\partial}{\partial s_i} f(\bb{s}) + \frac{1}{2} \sum_{i,j\in [d]} s_i (\ind_{\{i = j\}} - s_j) \frac{\partial^2}{\partial s_i \partial s_j} f(\bb{s}).
\end{equation}
\end{proposition}

\begin{proposition}[Pointwise variance]\label{prop:variance.pseudo}
Suppose that Assumptions~\ref{ass:1}, \ref{ass:2}, \ref{ass:5}, and \ref{ass:6} hold. For any $\bb{s}\in \mathrm{Int}(\mathcal{S}_d)$ such that $f(\bb{s})\in (0,\infty)$, we have
\[
\begin{aligned}
\Var(\widetilde{f}_{n,b}(\bb{s}))
&= \Var(\hat{f}^{\hspace{0.2mm}\mathrm{full}}_{n,b}(\bb{s})) + n^{-1} b^{-d/2} \psi(\bb{s}) f(\bb{s}) \zeta(\bb{s}) + \oo_{\bb{s}}(n^{-1} b^{-d/2}) \\
&= n^{-1} b^{-d/2} \psi(\bb{s}) f(\bb{s}) (1 + \zeta(\bb{s})) + \oo_{\bb{s}}(n^{-1} b^{-d/2}), \quad n\to \infty,
\end{aligned}
\]
where
\begin{equation}\label{eq:psi.zeta}
\psi(\bb{s}) = \left\{(4\pi)^d (1 - \|\bb{s}\|_1) \prod_{i=1}^d s_i\right\}^{-1/2}, \quad \zeta(\bb{s}) = \EE\left[\frac{1 - \pi(\bb{X}_1)}{\pi(\bb{X}_1)} \mid \bb{Y}_1 = \bb{s}\right].
\end{equation}
\end{proposition}

A consequence of Propositions~\ref{prop:bias.pseudo}--\ref{prop:variance.pseudo} is the following asymptotic expression for the MSE of the pseudo estimator. Assumption~\ref{ass:2} is not needed because Assumption~\ref{ass:3} contains it.

\begin{corollary}[Mean squared error]\label{cor:MSE.pseudo}
Suppose that Assumptions~\ref{ass:1}, \ref{ass:3}, \ref{ass:5}, and \ref{ass:6} hold. For any $\bb{s}\in \mathrm{Int}(\mathcal{S}_d)$ such that $f(\bb{s})\in (0,\infty)$, we have
\[
\begin{aligned}
\mathrm{MSE}[\widetilde{f}_{n,b}(\bb{s})]
&= \EE[(\widetilde{f}_{n,b}(\bb{s}) - f(\bb{s}))^2] \\
&= \Var(\widetilde{f}_{n,b}(\bb{s})) + \{\Bias[\widetilde{f}_{n,b}(\bb{s})]\}^2 \\
&= n^{-1} b^{-d/2} \psi(\bb{s}) f(\bb{s}) (1 + \zeta(\bb{s})) + b^2 \phi^2(\bb{s}) +\oo_{\bb{s}}(n^{-1} b^{-d/2}) + \oo(b^2), \quad n\to \infty.
\end{aligned}
\]
In particular, if $f(\bb{s}) \phi(\bb{s}) \neq 0$, the asymptotically optimal (pointwise) $b$, minimizing the two leading terms of the MSE, is given by
\[
b_{\mathrm{opt}}(\bb{s}) = n^{-2/(d+4)} \left[\frac{d}{4} \times \frac{\psi(\bb{s}) f(\bb{s}) (1 + \zeta(\bb{s}))}{\phi^2(\bb{s})}\right]^{2/(d+4)},
\]
and the corresponding MSE is
\[
\begin{aligned}
\mathrm{MSE}[\widetilde{f}_{n,b_{\mathrm{opt}}}]
&= n^{-4 / (d+4)} \left[\frac{1 + d/4}{(d/4)^{d/(d+4)}}\right] \frac{\big\{\psi(\bb{s}) f(\bb{s}) (1 + \zeta(\bb{s}))\big\}^{4 / (d+4)}}{\{\phi^2(\bb{s})\}^{-d / (d+4)}} + \oo_{\bb{s}}(n^{-4/(d+4)}).
\end{aligned}
\]
\end{corollary}

\begin{theorem}[Asymptotic normality]\label{thm:CLT.pseudo}
Suppose that Assumptions~\ref{ass:1}, \ref{ass:3}, \ref{ass:5}, and \ref{ass:6} hold. Choose $b\sim c n^{-2/(d+4)}$ as $n\to \infty$ for some $c\in (0,\infty)$. For any $\bb{s}\in \mathrm{Int}(\mathcal{S}_d)$ such that $f(\bb{s})\in (0,\infty)$, we have
\[
n^{1/2} b^{\hspace{0.2mm}d/4} \frac{\widetilde{f}_{n,b}(\bb{s}) - f(\bb{s}) - b \phi(\bb{s})}{\sqrt{\psi(\bb{s}) f(\bb{s}) (1 + \zeta(\bb{s}))}} \rightsquigarrow \mathcal{N}(0, 1).
\]
\end{theorem}

\subsection{Results for the feasible estimator \texorpdfstring{$\hat{f}_{n,b}$}{hat(f)\_n,b}}\label{sec:results.feasible}

\begin{proposition}[Pointwise bias]\label{prop:bias.feasible}
Suppose that Assumptions~\ref{ass:1}, \ref{ass:3}--\ref{ass:6}, and \ref{ass:B.1}--\ref{ass:B.3} hold. Choose $h\sim \kappa n^{-1/(p+4)}$ as $n\to \infty$ for some $\kappa\in (0,\infty)$. Uniformly for $\bb{s}\in \mathcal{S}_d$, we have
\[
\Bias(\hat{f}_{n,b}(\bb{s}))
= \Bias(\widetilde{f}_{n,b}(\bb{s})) + \OO(n^{-2/(p+4)})
= b \, \phi(\bb{s}) + \oo(b) + \OO(n^{-2/(p+4)}), \quad n\to \infty,
\]
where $\phi$ is defined in \eqref{eq:phi}.
\end{proposition}

\begin{proposition}[Pointwise variance]\label{prop:variance.feasible}
Suppose that Assumptions~\ref{ass:1}, \ref{ass:2}, \ref{ass:4}--\ref{ass:6}, and \ref{ass:B.1}--\ref{ass:B.3} hold. Choose $b\sim c n^{-2/(d+4)}$ and $h\sim \kappa n^{-1/(p+4)}$ as $n\to \infty$. For any $\bb{s}\in \mathrm{Int}(\mathcal{S}_d)$ such that $f(\bb{s})\in (0,\infty)$, we have
\[
\begin{aligned}
\Var(\hat{f}_{n,b}(\bb{s}))
&= \Var(\widetilde{f}_{n,b}(\bb{s})) - n^{-1} \xi(\bb{s}) + \OO(n^{-4/(p+4)}) + \OO(n^{-2/(p+4)} n^{-2/(d+4)}) \\
&= n^{-1} b^{-d/2} \psi(\bb{s}) f(\bb{s}) (1 + \zeta(\bb{s})) - n^{-1} \xi(\bb{s}) \\
&\quad+ \oo_{\bb{s}}(n^{-4/(d+4)}) + \OO(n^{-4/(p+4)}) + \OO(n^{-2/(p+4)} n^{-2/(d+4)}),
\end{aligned}
\]
where $\psi$ and $\zeta$ are defined in \eqref{eq:psi.zeta}, and
\[
\xi(\bb{s}) = \EE\left[\frac{1 - \pi(\bb{X}_1)}{\pi(\bb{X}_1)} \big\{\EE[\kappa_{\bb{s},b}(\bb{Y}_1) \mid \bb{X}_1]\big\}^2\right].
\]
\end{proposition}

\begin{corollary}[Mean squared error]\label{cor:MSE.feasible}
Suppose that Assumptions~\ref{ass:1}, \ref{ass:2}, \ref{ass:4}--\ref{ass:6}, and \ref{ass:B.1}--\ref{ass:B.3} hold. Choose $b\sim c n^{-2/(d+4)}$ and $h\sim \kappa n^{-1/(p+4)}$ as $n\to \infty$ for some $c,\kappa\in (0,\infty)$. For any $\bb{s}\in \mathrm{Int}(\mathcal{S}_d)$ such that $f(\bb{s})\in (0,\infty)$, we have
\[
\begin{aligned}
\mathrm{MSE}[\hat{f}_{n,b}(\bb{s})]
&= \EE[(\hat{f}_{n,b}(\bb{s}) - f(\bb{s}))^2] \\
&= \Var(\hat{f}_{n,b}(\bb{s})) + \{\Bias[\hat{f}_{n,b}(\bb{s})]\}^2 \\
&= n^{-1} b^{-d/2} \psi(\bb{s}) f(\bb{s}) (1 + \zeta(\bb{s})) - n^{-1} \xi(\bb{s}) + b^2 \phi^2(\bb{s}) \\
&\quad+ \oo_{\bb{s}}(n^{-4/(d+4)}) + \OO(n^{-4/(p+4)}) + \OO(n^{-2/(p+4)} n^{-2/(d+4)}), \quad n\to \infty.
\end{aligned}
\]
\end{corollary}

\begin{theorem}[Asymptotic normality]\label{thm:CLT.feasible}
Suppose that Assumptions~\ref{ass:1}, \ref{ass:3}--\ref{ass:6}, and \ref{ass:B.1}--\ref{ass:B.3} hold. Choose $b\sim c n^{-2/(d+4)}$ and $h\sim \kappa n^{-1/(p+4)}$ as $n\to \infty$ for some $c,\kappa\in (0,\infty)$. If $p < d$, then for any $\bb{s}\in \mathrm{Int}(\mathcal{S}_d)$ such that $f(\bb{s})\in (0,\infty)$, we have
\[
n^{1/2} b^{\hspace{0.2mm}d/4} \frac{\hat{f}_{n,b}(\bb{s}) - f(\bb{s}) - b \phi(\bb{s})}{\sqrt{\psi(\bb{s}) f(\bb{s}) (1 + \zeta(\bb{s}))}} \rightsquigarrow \mathcal{N}(0, 1).
\]
\end{theorem}

\begin{remark}\label{rem:dimensionality}
In Theorem~\ref{thm:CLT.feasible}, the condition $p < d$ ensures that the error introduced by nonparametrically estimating the propensity scores does not overpower the error of the density estimation itself. In our framework, the Nadaraya--Watson estimator for the propensity score (the nuisance parameter) converges at a rate of $\mathcal{O}(n^{-4/(p+4)})$, while the Dirichlet KDE converges at a rate of $\mathcal{O}(n^{-4/(d+4)})$. For the uncertainty of the propensity score estimation to be asymptotically negligible at the first order, its convergence must strictly dominate the density estimation error: $n^{-4/(p+4)} = \mathrm{o}(n^{-4/(d+4)})$, which algebraically simplifies to $p < d$. If $p \geq d$, the curse of dimensionality inflates the variance of the inverse probability weights, causing the nuisance parameter error to dominate and invalidating standard asymptotic normality results. One can still obtain the asymptotic normality when $p \geq d$ by imposing stricter H\"older-type smoothness conditions on $\pi$ and the covariate density $g$, possibly coupled with the use of higher-order kernels. This would accelerate the convergence rate of the Nadaraya--Watson estimator; see, e.g., \citet{Alwadeai2026} for such a treatment in the context of beta kernel density estimation with missing data. We have chosen to omit this variant here to avoid overly technical assumptions and maintain the clarity of our presentation. In practice, high-dimensional covariate settings often require dimension reduction techniques or parametric working models for the missingness mechanism to restore efficiency; see, for example, \citet{Hu2010}, \citet{Wang2018}, and \citet{Ertefaie2020}.
\end{remark}

\section{Simulation results}\label{sec:simulation}

This section presents a Monte Carlo study assessing the finite-sample performance of the proposed feasible IPW Dirichlet KDE $\hat{f}_{n,b}$. Section~\ref{sec:models.setup} describes the data-generating mechanism for $(\bb{X},\bb{Y})$, the MAR missingness model, and the propensity-score estimation procedure. Section~\ref{sec:bandwidth} details the IPW-adapted least-squares cross-validation criterion used to select the Dirichlet bandwidth $b$. Section~\ref{sec:performance} defines the performance metric and summarizes results across models, sample sizes, and missing rates. Finally, Section~\ref{sec:comparison} compares the proposed estimator with IPW alternatives based on additive and isometric log-ratio transformations, and studies the effect of sample size and missingness.

\subsection{Models and setup}\label{sec:models.setup}

The simulation involves a continuous simplex-valued response $\bb{Y}\in\mathcal{S}_2$ and a bivariate covariate vector $\bb{X}\in\R^2$. For each Monte Carlo replication, we generate an iid sample $\{(\bb{X}_i,\bb{Y}_i,\delta_i)\}_{i=1}^n$ as follows:
\begin{itemize}\setlength\itemsep{0em}
\item \textbf{Generate the response.} Sample $\bb{Y}_i$ from a two-component Dirichlet mixture distribution $F_{\bb{Y}}$:
\begin{itemize}\setlength\itemsep{0em}
\item \textbf{Model~I}: $0.4 \cdot \mathrm{Dirichlet}(1.3, 1.6, 1) + 0.6 \cdot \mathrm{Dirichlet}(1.7, 1.2, 2.5)$.
\item \textbf{Model~II}: $0.4 \cdot \mathrm{Dirichlet}(4, 1, 2) + 0.6 \cdot \mathrm{Dirichlet}(1, 3, 2)$.
\end{itemize}
\item \textbf{Generate the covariates.} Set $\bb{X}_i = \rho \, \bb{Y}_i + \sqrt{1 - \rho^{2}} \, \bb{X}^{*}_i$, where $-1 < \rho < 1$ controls the strength of association between $\bb{X}_i$ and $\bb{Y}_i$, and $\bb{X}^{*}_i \sim \mathcal{N}_2(\bb{0}_2,I_2)$ is independent of $\bb{Y}_i$.
\item \textbf{Generate MAR missingness.} Sample the missingness indicator $\delta_i$ under a MAR mechanism through the logistic propensity score model
\begin{equation}\label{eq:logistic.model}
\pi(\bb{X}_i) = \PP\left(\delta_i=1\mid \bb{X}_i\right)
= \mathrm{logit}^{-1}\big(\beta_0 + \bb{\beta}_1^{\top}\bb{X}_i\big),
\end{equation}
where $\beta_0 \in \R$ and $\bb{\beta}_1 \in \R^2$, and then $\delta_i \sim \mathrm{Bernoulli}(\pi(\bb{X}_i)), ~ i\in [n]$. In each scenario, $(\beta_0,\bb{\beta}_1)$ is calibrated so that the marginal proportion of missing responses matches the targeted missing rate.
\end{itemize}

Across both models, we consider sample sizes $n\in\{100,200,400,800\}$ and a range of missingness proportions, and we report results over $1000$ Monte Carlo replications for each configuration. The propensity score $\pi(\bb{X}_i)$ is treated as unknown for the purpose of calculating density estimates, so it is estimated using the Nadaraya--Watson estimator $\hat{\pi}_i(\bb{X}_{1:n})$ from \eqref{eq:NW}, with $K^*$ taken as the standard bivariate Gaussian kernel. The bandwidth $h$ in $K_h^*$ is selected by Silverman's rule of thumb, $h = 1.06\,\sigma\,n^{-1/(p+4)}$, where $p=2$ and $\sigma$ is the average empirical standard deviation across the coordinates of $\bb{X}$.

The induced MAR observation patterns for the two mixture models are illustrated in Figure~\ref{fig:MAR}.
\begin{figure}[H]
\includegraphics[trim=1.3cm 12cm 1.5cm 2.2cm, clip, width=0.49\textwidth]{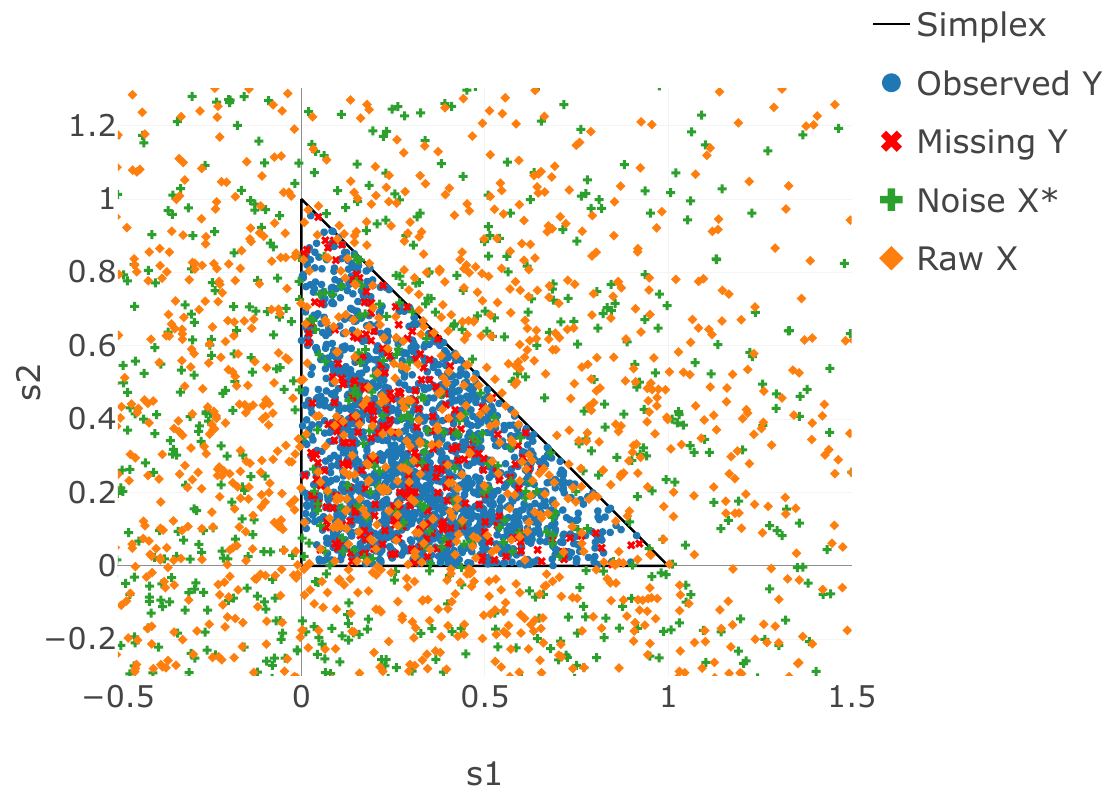}
\hfill
\includegraphics[trim=1.3cm 12cm 1.5cm 2.2cm, clip, width=0.49\textwidth]{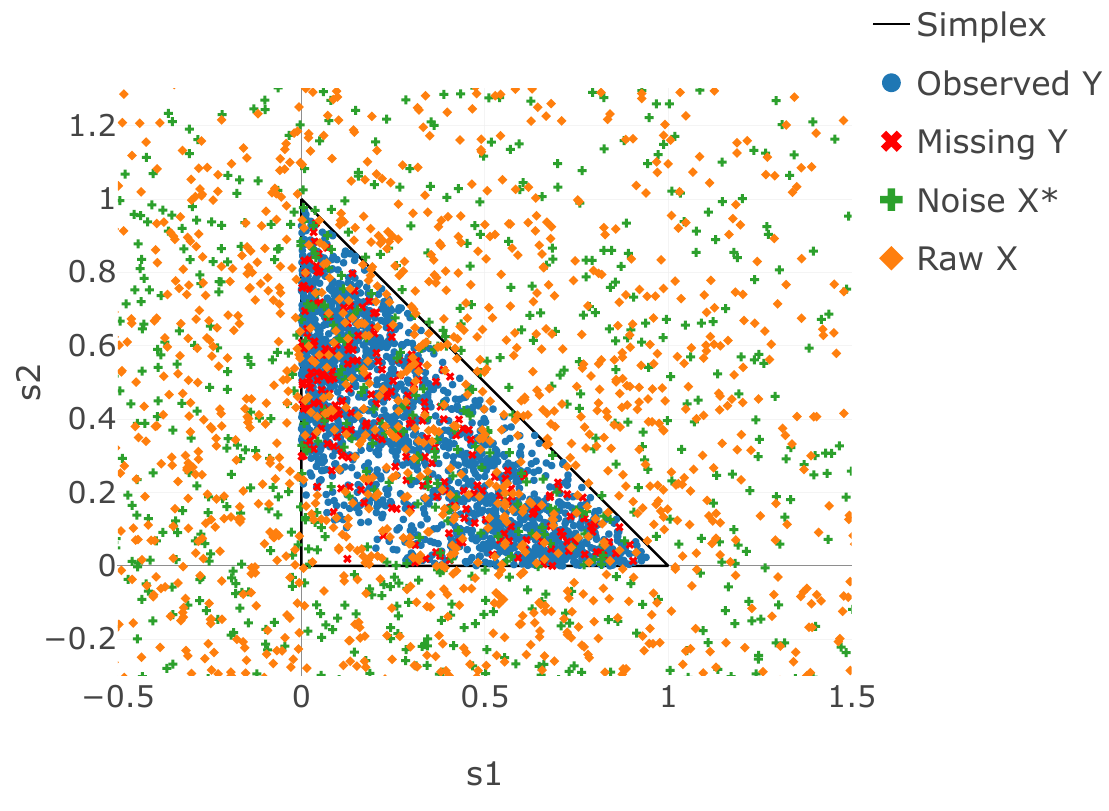}
\caption{Visualization of the MAR mechanisms for Model~I (left panel) and Model~II (right panel).}
\label{fig:MAR}
\end{figure}

Figure~\ref{fig:model.1} provides a visual comparison between the target density of Model I and the associated Dirichlet kernel density estimate under missingness, with $n=2000$ and a $10\%$ missing rate. The analogous comparison for Model~II is displayed in Figure~\ref{fig:model.2}. In both cases, the contour plots are obtained by evaluating the densities on a deterministic interior grid of $\mathcal{S}_2$ with resolution $\texttt{res}=300$ and tolerance $\varepsilon=0.01$ (see Section~\ref{sec:bandwidth} for the grid construction).

\begin{figure}[H]
\includegraphics[trim=0.3cm 1.8cm 0.5cm 1.5cm, clip, width=0.49\textwidth]{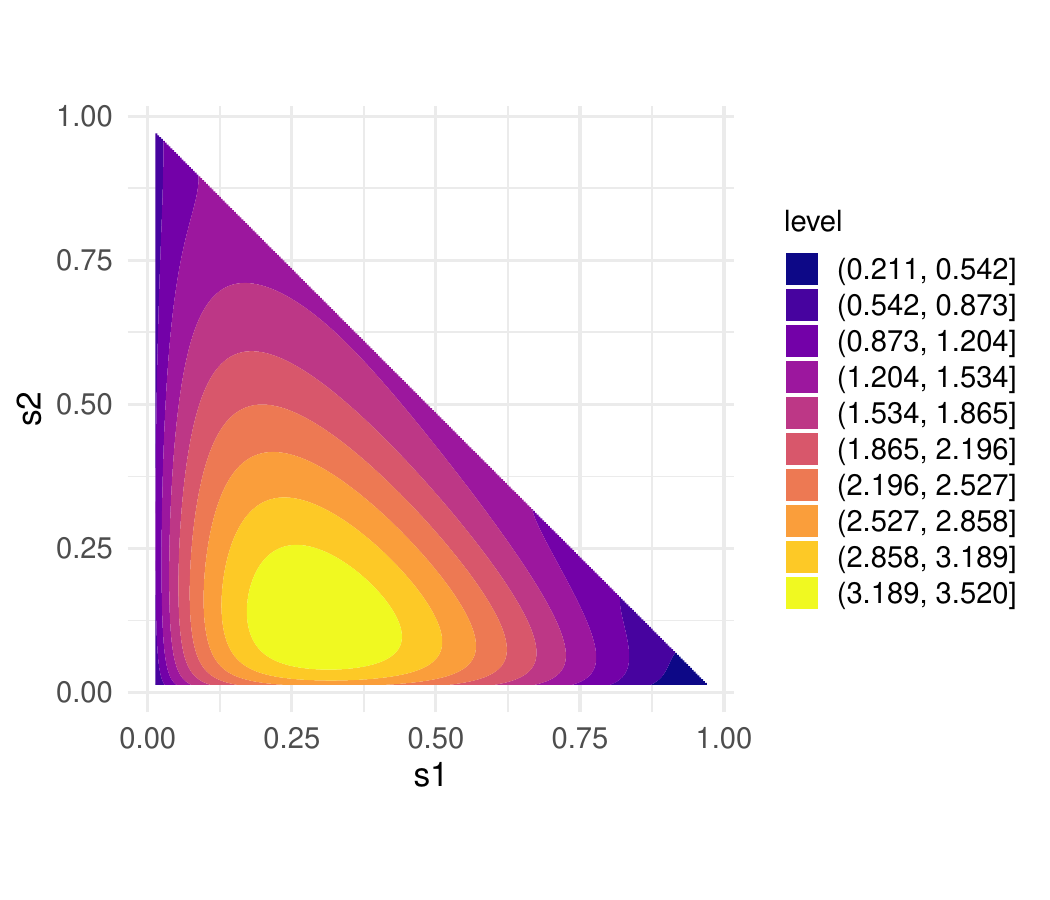}
\hfill
\includegraphics[trim=0.3cm 1.8cm 0.5cm 1.5cm, clip, width=0.49\textwidth]{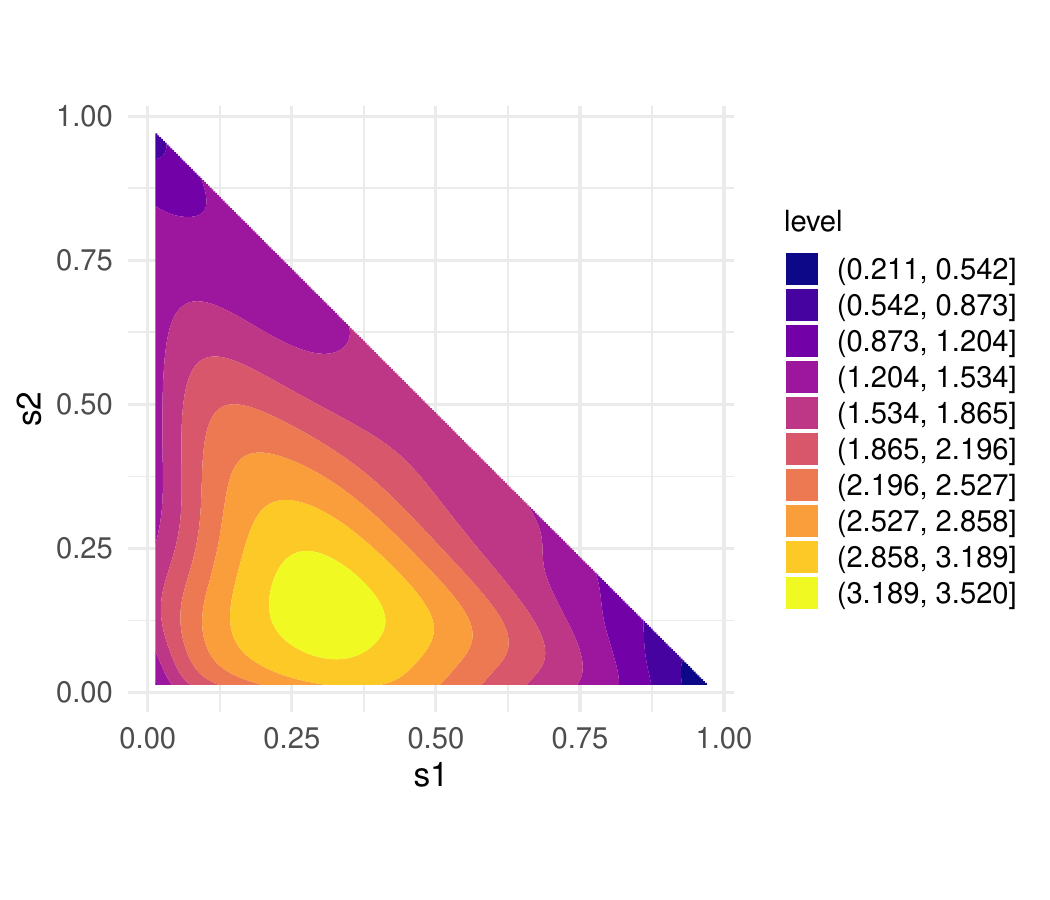}
\caption{Contour plots of the Model~I target density $f$ (left panel) and the associated Dirichlet kernel density estimate $\hat{f}_{n,0.05}$ (right panel), with a sample size $n = 2000$ and a $10\%$ missing rate.}
\label{fig:model.1}
\end{figure}

\begin{figure}[H]
\includegraphics[trim=0.3cm 2cm 0.5cm 2cm, clip, width=0.49\textwidth]{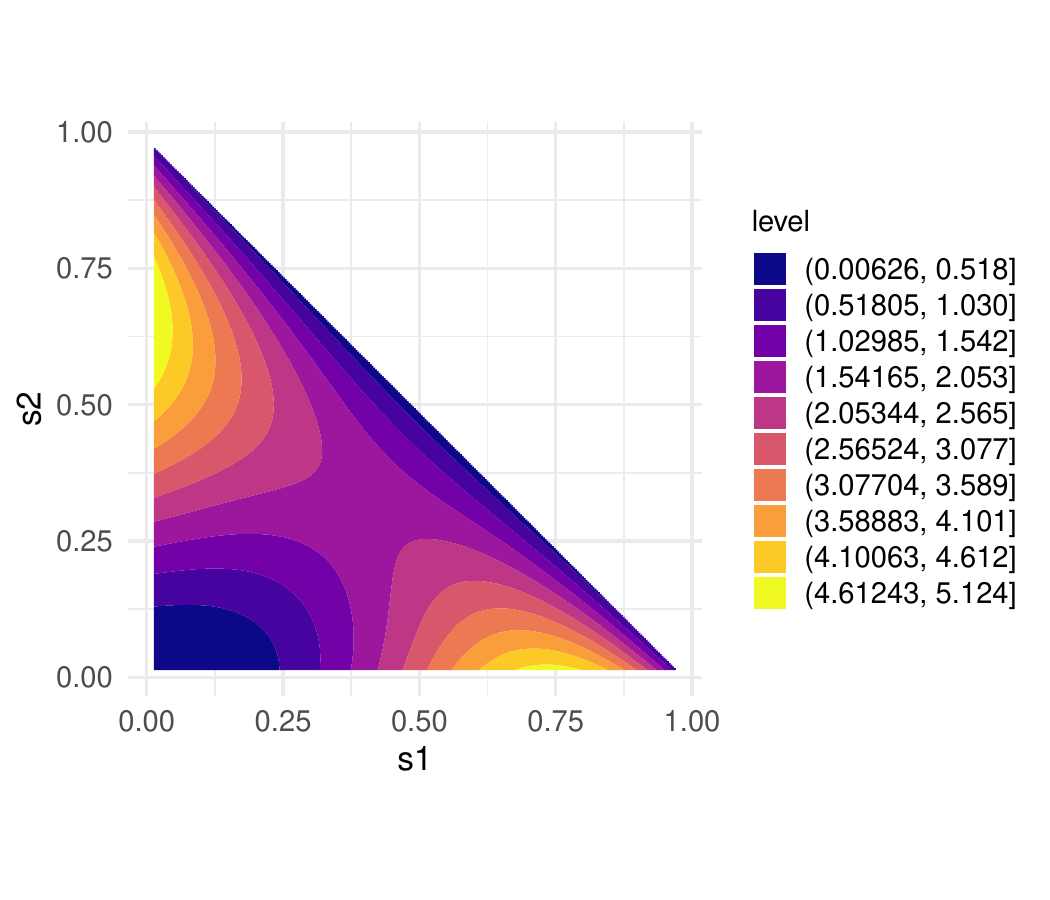}
\hfill
\includegraphics[trim=0.3cm 2cm 0.5cm 2cm, clip, width=0.49\textwidth]{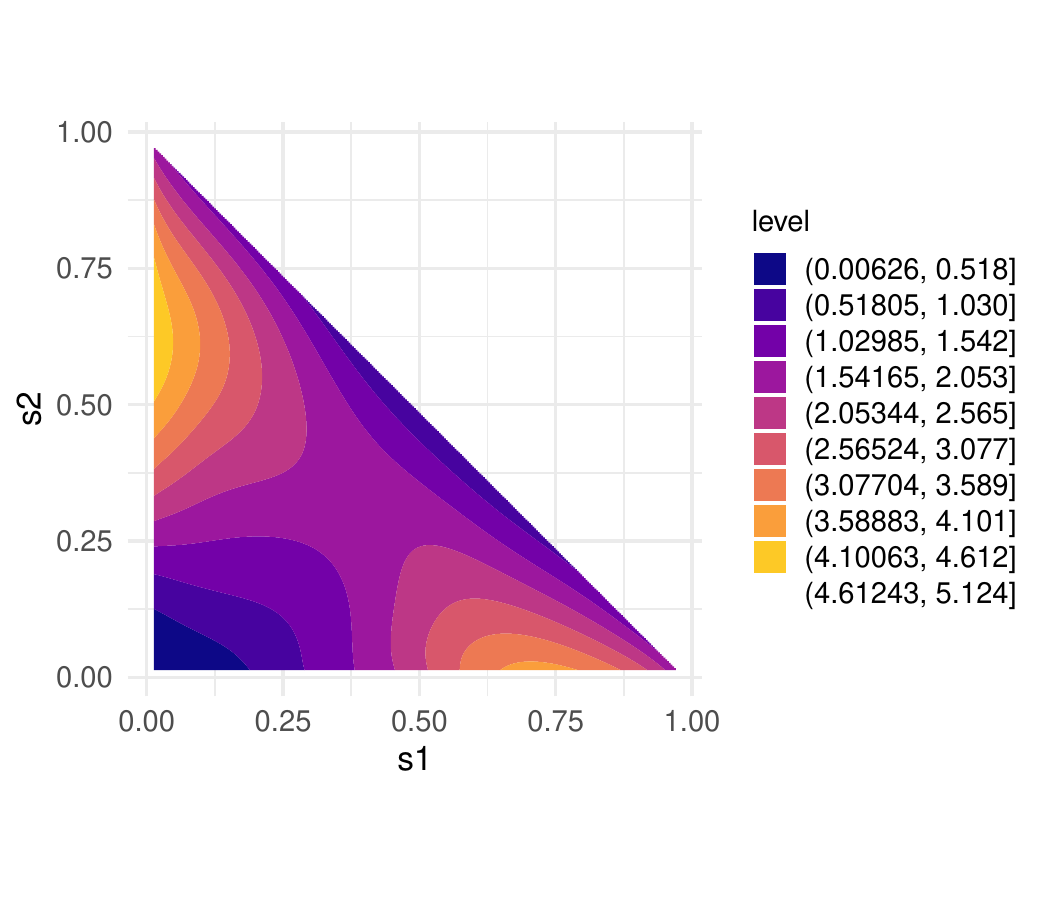}
\caption{Contour plots of the Model~II target density $f$ (left panel) and the associated Dirichlet kernel density estimate $\hat{f}_{n,0.05}$ (right panel), with a sample size $n = 2000$ and a $10\%$ missing rate.}
\label{fig:model.2}
\end{figure}

\subsection{Bandwidth selection}\label{sec:bandwidth}

In each replication, the Dirichlet bandwidth $b$ is selected by least-squares cross-validation (LSCV) adapted to the IPW setting. The criterion targets the integrated squared error (ISE) while accounting for missingness through inverse probability weights, and it includes a leave-one-out correction to prevent in-sample bias. For each $b$ in a candidate grid $\mathcal{B}\subseteq (0,1)$, we compute
\begin{equation}\label{eq:LSCV_b_simulation}
\mathrm{LSCV}(b) = \int_{\mathcal{S}_2} \hat{f}_{n,b}^2(\bb{s})\,\rd\bb{s}
- \frac{2}{n} \sum_{i=1}^n \frac{\delta_i}{\hat{\pi}_i(\bb{X}_{1:n})}\,\hat{f}_{n,b}^{\, (-i)}(\bb{Y}_i),
\end{equation}
where $\smash{\hat{f}_{n,b}^{\, (-i)}}$ denotes the estimator in \eqref{eq:feasible} without the $i$th observation.

Both the integral in \eqref{eq:LSCV_b_simulation} and the ISE approximation \eqref{eq:ISE.approximation} in Section~\ref{sec:performance} are evaluated on deterministic interior grids of $\mathcal{S}_2$. For a resolution parameter $\texttt{res}\geq 4$ and a tolerance $\varepsilon\in(0,1/3)$, let $H = \texttt{res}-1$ and define
\[
\mathcal{G}(\texttt{res},\varepsilon) = \Big\{(s_1,s_2)\in\mathrm{Int}(\mathcal{S}_2): s_{\ell} = \varepsilon+(1-3\varepsilon)\frac{i_{\ell}}{H},~ \ell\in\{1,2,3\},~ i_1,i_2,i_3\in \N,~ i_1+i_2+i_3 = H \Big\}.
\]
This construction yields
\[
M = \frac{(H-1)(H-2)}{2}
\]
grid points, all bounded away from the simplex boundary by $\varepsilon$. In our \textsf{R} implementation, we take $\varepsilon=0.01$ throughout. For the LSCV integral approximation we use $\texttt{res}=40$ (so $H=39$ and $M=703$ grid points).

\medskip

The first term in \eqref{eq:LSCV_b_simulation} can be written as
\[
\int_{\mathcal{S}_2} \hat{f}_{n,b}^2(\bb{s})\,\rd\bb{s}
= \frac{1}{n^2}\sum_{i=1}^n\sum_{j=1}^n \frac{\delta_i\,\delta_j}{\hat{\pi}_i(\bb{X}_{1:n}) \hat{\pi}_j(\bb{X}_{1:n})}
\int_{\mathcal{S}_2} \kappa_{\bb{s},b}(\bb{Y}_i) \kappa_{\bb{s},b}(\bb{Y}_j)\,\rd\bb{s},
\]
and, in practice, it is approximated numerically on $\mathcal{G}(\texttt{res},\varepsilon)$ via the Riemann-sum approximation
\[
\int_{\mathcal{S}_2} \hat{f}_{n,b}^2(\bb{s})\,\rd\bb{s}
\approx |\mathcal{S}_2| \frac{1}{M} \sum_{m=1}^M \hat{f}_{n,b}^2(\bb{s}_m)
= \frac{1}{2M}\sum_{m=1}^M \hat{f}_{n,b}^2(\bb{s}_m),
\]
where $\bb{s}_1,\ldots,\bb{s}_M\in\mathcal{G}(\texttt{res},\varepsilon)$ and $|\mathcal{S}_2|=1/2$ is the area of the two-dimensional simplex.

The selected bandwidth is
\begin{equation}\label{eq:b.star}
b^{\star} = \arg\min_{b \in \mathcal{B}} \mathrm{LSCV}(b).
\end{equation}
In the simulations reported in Tables~\ref{tab:results_LSCV_Model1}--\ref{tab:results_LSCV_Model2} below, we take the candidate set
\[
\mathcal{B} = \{0.01,0.02,\ldots,0.35\},
\]
i.e., $35$ equally spaced values from $0.01$ to $0.35$.

\subsection{Performance evaluation}\label{sec:performance}

The performance of the estimator $\hat{f}_{n,b^{\star}}$ is assessed by the integrated squared error (ISE),
\[
\mathrm{ISE}(\hat{f}_{n,b^{\star}}) = \int_{\mathcal{S}_2} \left(\hat{f}_{n,b^{\star}}(\bb{s})- f(\bb{s})\right)^2 \, \rd\bb{s}.
\]
In practice, we approximate it numerically using a deterministic interior grid $\bb{s}_1,\ldots,\bb{s}_M \in \mathrm{Int}(\mathcal{S}_2)$:
\begin{equation}\label{eq:ISE.approximation}
\mathrm{ISE}(\hat{f}_{n,b^{\star}}) \approx \frac{1}{2M} \sum_{m=1}^M  \big\{\hat{f}_{n,b^{\star}}(\bb{s}_m) - f(\bb{s}_m)\big\}^2.
\end{equation}
In our \textsf{R} implementation, we use the grid construction described in Section~\ref{sec:bandwidth} with $\varepsilon=0.01$ and $\texttt{res}=300$ (so $H=299$ and $M=44{,}253$ evaluation points) to evaluate ISE.

For each configuration (model, sample size, and missing rate), we compute ISE over $1000$ Monte Carlo replications and summarize the resulting distribution through its mean, median, standard deviation (SD), and interquartile range (IQR). The reported ``Mean~$b^{\star}$'' corresponds to the average LSCV-selected bandwidth across replications.

Summary results for Models I and II are reported in Tables~\ref{tab:results_LSCV_Model1} and \ref{tab:results_LSCV_Model2}, respectively.

\begin{table}[H]
\renewcommand{\arraystretch}{1.1}
\centering
\begin{tabular}{rrrrrrr}
\hline
n & Missing Rate (\%) & Mean ISE & Median ISE & SD ISE & IQR ISE & Mean $b^{\star}$ \\
\hline
100 & 5 & 0.1484 & 0.1400 & 0.0771 & 0.0723 & 0.2142 \\
100 & 10 & 0.1520 & 0.1448 & 0.0844 & 0.0804 & 0.2228 \\
100 & 20 & 0.1590 & 0.1512 & 0.0841 & 0.0741 & 0.2289 \\
100 & 40 & 0.2070 & 0.1925 & 0.1309 & 0.0920 & 0.2517 \\
200 & 5 & 0.1087 & 0.0997 & 0.0485 & 0.0660 & 0.1741 \\
200 & 10 & 0.1099 & 0.1016 & 0.0480 & 0.0637 & 0.1776 \\
200 & 20 & 0.1176 & 0.1098 & 0.0531 & 0.0715 & 0.1842 \\
200 & 40 & 0.1453 & 0.1434 & 0.0620 & 0.0841 & 0.2105 \\
400 & 5 & 0.0743 & 0.0684 & 0.0316 & 0.0405 & 0.1249 \\
400 & 10 & 0.0774 & 0.0710 & 0.0343 & 0.0468 & 0.1281 \\
400 & 20 & 0.0855 & 0.0807 & 0.0345 & 0.0440 & 0.1388 \\
400 & 40 & 0.1049 & 0.1001 & 0.0439 & 0.0589 & 0.1555 \\
800 & 5 & 0.0547 & 0.0506 & 0.0214 & 0.0271 & 0.1002 \\
800 & 10 & 0.0573 & 0.0547 & 0.0223 & 0.0291 & 0.1031 \\
800 & 20 & 0.0618 & 0.0586 & 0.0244 & 0.0305 & 0.1072 \\
800 & 40 & 0.0763 & 0.0740 & 0.0308 & 0.0409 & 0.1186 \\
\hline
\end{tabular}
\caption{The mean, median, standard deviation, and interquartile range of 1000 ISEs in Model I for the IPW Dirichlet KDE, with sample sizes $n\in \{100, 200, 400, 800\}$ and missing rates of $5\%$, $10\%$, $20\%$, and $40\%$.}
\label{tab:results_LSCV_Model1}
\end{table}

\begin{table}[H]
\renewcommand{\arraystretch}{1.1}
\centering
\begin{tabular}{rrrrrrr}
\hline
n & Missing Rate (\%) & Mean ISE & Median ISE & SD ISE & IQR ISE & Mean $b^{\star}$ \\
\hline
100 & 5 & 0.2546 & 0.2204 & 0.1523 & 0.1462 & 0.0616 \\
100 & 10 & 0.2713 & 0.2313 & 0.1648 & 0.1461 & 0.0637 \\
100 & 20 & 0.3004 & 0.2612 & 0.1762 & 0.1509 & 0.0695 \\
100 & 40 & 0.3742 & 0.3317 & 0.1992 & 0.1799 & 0.0906 \\
200 & 5 & 0.1663 & 0.1479 & 0.0856 & 0.0802 & 0.0405 \\
200 & 10 & 0.1746 & 0.1578 & 0.0859 & 0.0866 & 0.0414 \\
200 & 20 & 0.1942 & 0.1710 & 0.0982 & 0.0949 & 0.0449 \\
200 & 40 & 0.2511 & 0.2225 & 0.1199 & 0.1199 & 0.0528 \\
400 & 5 & 0.1058 & 0.0943 & 0.0473 & 0.0512 & 0.0283 \\
400 & 10 & 0.1111 & 0.0995 & 0.0490 & 0.0545 & 0.0288 \\
400 & 20 & 0.1245 & 0.1109 & 0.0535 & 0.0574 & 0.0312 \\
400 & 40 & 0.1627 & 0.1525 & 0.0639 & 0.0646 & 0.0370 \\
800 & 5 & 0.0703 & 0.0638 & 0.0264 & 0.0376 & 0.0200 \\
800 & 10 & 0.0736 & 0.0666 & 0.0273 & 0.0380 & 0.0208 \\
800 & 20 & 0.0840 & 0.0765 & 0.0307 & 0.0403 & 0.0217 \\
800 & 40 & 0.1103 & 0.1021 & 0.0373 & 0.0455 & 0.0253 \\
\hline
\end{tabular}
\caption{The mean, median, standard deviation, and interquartile range of 1000 ISEs in Model II for the IPW Dirichlet KDE, with sample sizes $n\in \{100, 200, 400, 800\}$ and missing rates of $5\%$, $10\%$, $20\%$, and $40\%$.}
\label{tab:results_LSCV_Model2}
\end{table}

Across both models, Tables~\ref{tab:results_LSCV_Model1}--\ref{tab:results_LSCV_Model2} exhibit two consistent trends: for a fixed missing rate, the mean, median, SD, and IQR of the $1000$ ISEs decrease as the sample size increases; for a fixed $n$, these summary statistics increase as the missing rate grows. These patterns are illustrated in Figures~\ref{fig:ise.metrics.model1} and \ref{fig:ise.metrics.model2}.

\newpage
\begin{figure}[H]
\centering
\includegraphics[width=0.75\textwidth]{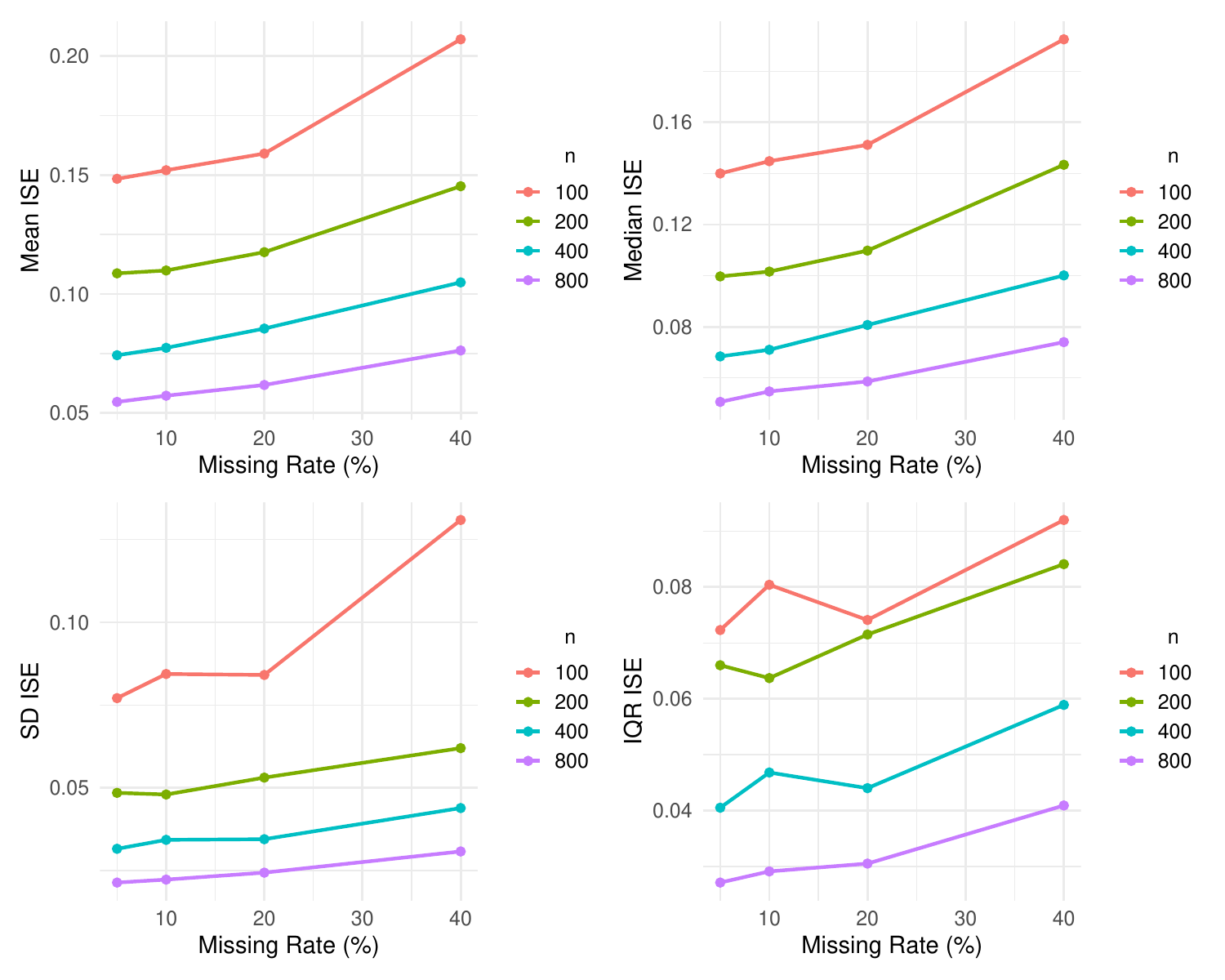}
\vspace{-2mm}
\caption{Mean, median, standard deviation, and interquartile range of $1000$ ISEs in Model~I for the IPW Dirichlet KDE as a function of the proportion of missing data, shown for four sample sizes $n\in \{100, 200, 400, 800\}$.}
\label{fig:ise.metrics.model1}
\end{figure}

\begin{figure}[H]
\centering
\includegraphics[width=0.75\textwidth]{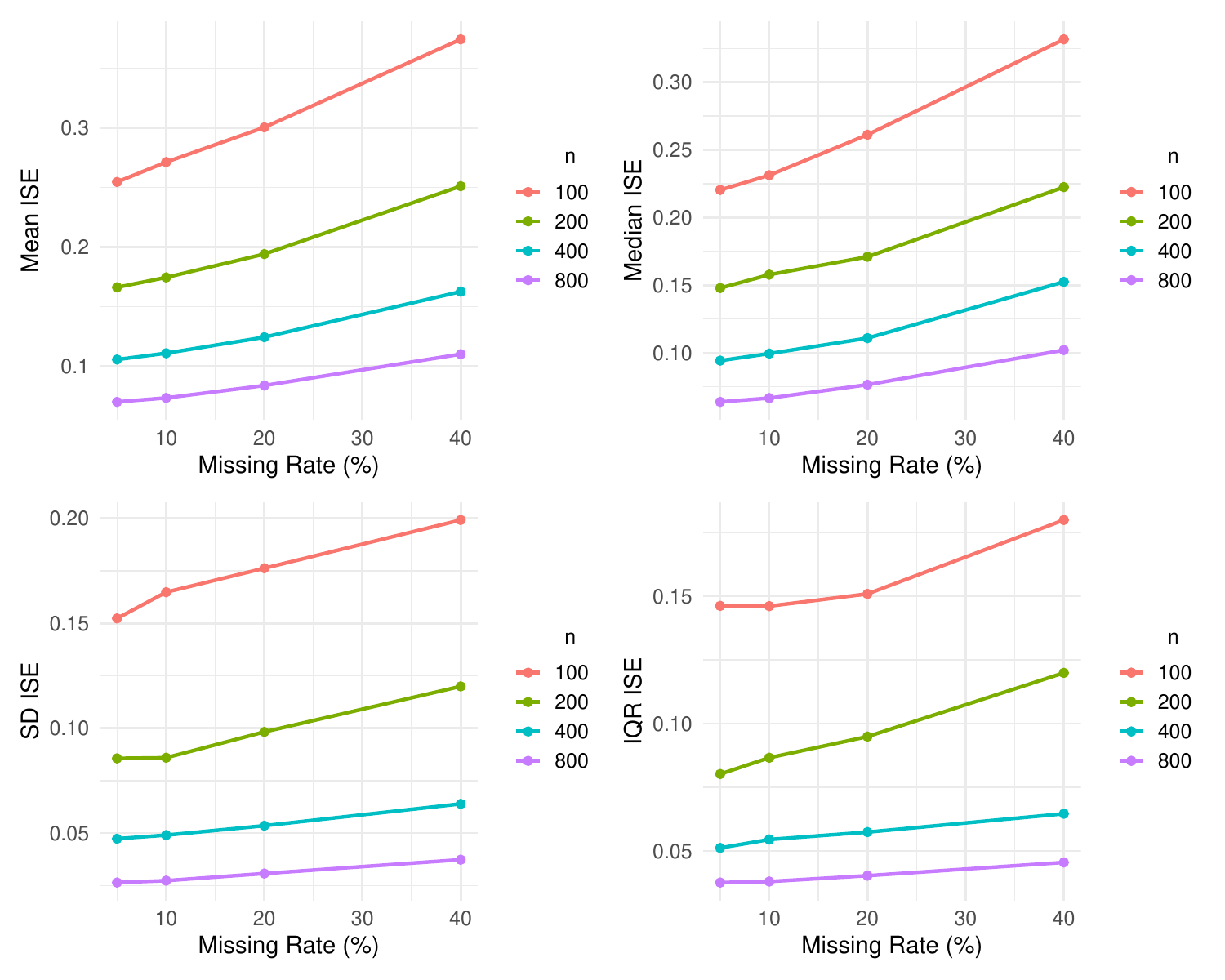}
\vspace{-2mm}
\caption{Mean, median, standard deviation, and interquartile range of $1000$ ISEs in Model~II for the IPW Dirichlet KDE as a function of the proportion of missing data, shown for four sample sizes $n\in \{100, 200, 400, 800\}$.}
\label{fig:ise.metrics.model2}
\end{figure}

\subsection{Comparative study}\label{sec:comparison}

We complement the previous analysis by comparing the proposed IPW Dirichlet KDE with two log-ratio-based alternatives, obtained by performing kernel density estimation in an unconstrained Euclidean space and transforming the estimate back to the simplex with the appropriate Jacobian correction.

Let $\bb{Y}_i = (Y_{i1},Y_{i2}) \in \mathcal{S}_2$ denote the $i$-th compositional observation on the two-dimensional simplex, and set $Y_{i3}=1-\|\bb{Y}_i\|_1$. To construct alternative smoothing strategies, we employ the additive log-ratio (alr) transformation and the isometric log-ratio (ilr) transformation, which map $\mathrm{Int}(\mathcal{S}_2)$ onto $\R^2$:
\[
\begin{aligned}
\mathrm{alr}(\bb{Y}_i) &= (\log(Y_{i1}/Y_{i3}),\log(Y_{i2}/Y_{i3})), \\
\mathrm{ilr}(\bb{Y}_i) &= ((1/\sqrt{2}) \log(Y_{i1}/Y_{i2}), (1/\sqrt{6}) \log(Y_{i1}Y_{i2}/Y^2_{i3})).
\end{aligned}
\]

The inverse transformations are smooth. For a generic point $\bb{s}=(s_1,s_2)\in \mathrm{Int}(\mathcal{S}_2)$, set $s_3=1-\|\bb{s}\|_1$. The Jacobian determinants of the $\mathrm{alr}$ and $\mathrm{ilr}$ mappings are given by
\[
\left| \det\left( \frac{\partial \, \mathrm{alr}(\bb{s})}{\partial \bb{s}} \right) \right| = \frac{1}{s_1 s_2 s_3}, \qquad
\left| \det\left( \frac{\partial \, \mathrm{ilr}(\bb{s})}{\partial \bb{s}} \right) \right| = \frac{1}{\sqrt{3} \, s_1 s_2 s_3},
\]
see, e.g., \cite[p.~115~and~p.~120]{PawlowskyGlahn2015Modeling}. Let $K_h(\cdot)$ denote the centered bivariate Gaussian kernel of covariance $h I_2$. The IPW-alr and IPW-ilr KDEs on the simplex (when propensity scores are estimated) are respectively defined as
\[
\begin{aligned}
\hat{f}_{\mathrm{alr}}(\bb{s})
&= \frac{1}{n} \sum_{i=1}^n \frac{\delta_i}{\hat{\pi}_i(\bb{X}_{1:n})}
K_h \big(\mathrm{alr}(\bb{s}) - \mathrm{alr}(\bb{Y}_i)\big) \, \left| \det\left( \frac{\partial \, \mathrm{alr}(\bb{s})}{\partial \bb{s}} \right) \right|, \\
\hat{f}_{\mathrm{ilr}}(\bb{s})
&= \frac{1}{n} \sum_{i=1}^n \frac{\delta_i}{\hat{\pi}_i(\bb{X}_{1:n})}
K_h \big(\mathrm{ilr}(\bb{s}) - \mathrm{ilr}(\bb{Y}_i)\big) \, \left| \det\left( \frac{\partial \, \mathrm{ilr}(\bb{s})}{\partial \bb{s}} \right) \right|.
\end{aligned}
\]
The unweighted versions of these estimators were introduced by \citet{doi:10.2307/2347365} (alr) and \citet{doi:10.1016/j.cageo.2009.12.011} (ilr), respectively.

The models, setup, bandwidth selection, and performance metric remain exactly as described in Sections~\ref{sec:models.setup}, \ref{sec:bandwidth}, and \ref{sec:performance}, with $\mathrm{ISE}(\hat{f}_{\mathrm{trans}})$, for $\mathrm{trans}\in \{\mathrm{alr},\mathrm{ilr}\}$, defined in the obvious way.

\subsubsection{Impact of the sample size}

Increasing the sample size $n$ leads to systematic improvements in estimator performance. Both Dirichlet and log-ratio-based estimators exhibit decreasing median ISE and reduced dispersion as $n$ increases from $100$ to $800$.

This effect is illustrated in Figure~\ref{fig:comparison.n}, which reports boxplots of $1000$ ISEs across the four sample sizes under $20\%$ missingness. For all sample sizes and models, the IPW Dirichlet KDE performs best.

\begin{figure}[H]
\includegraphics[width=0.49\textwidth]{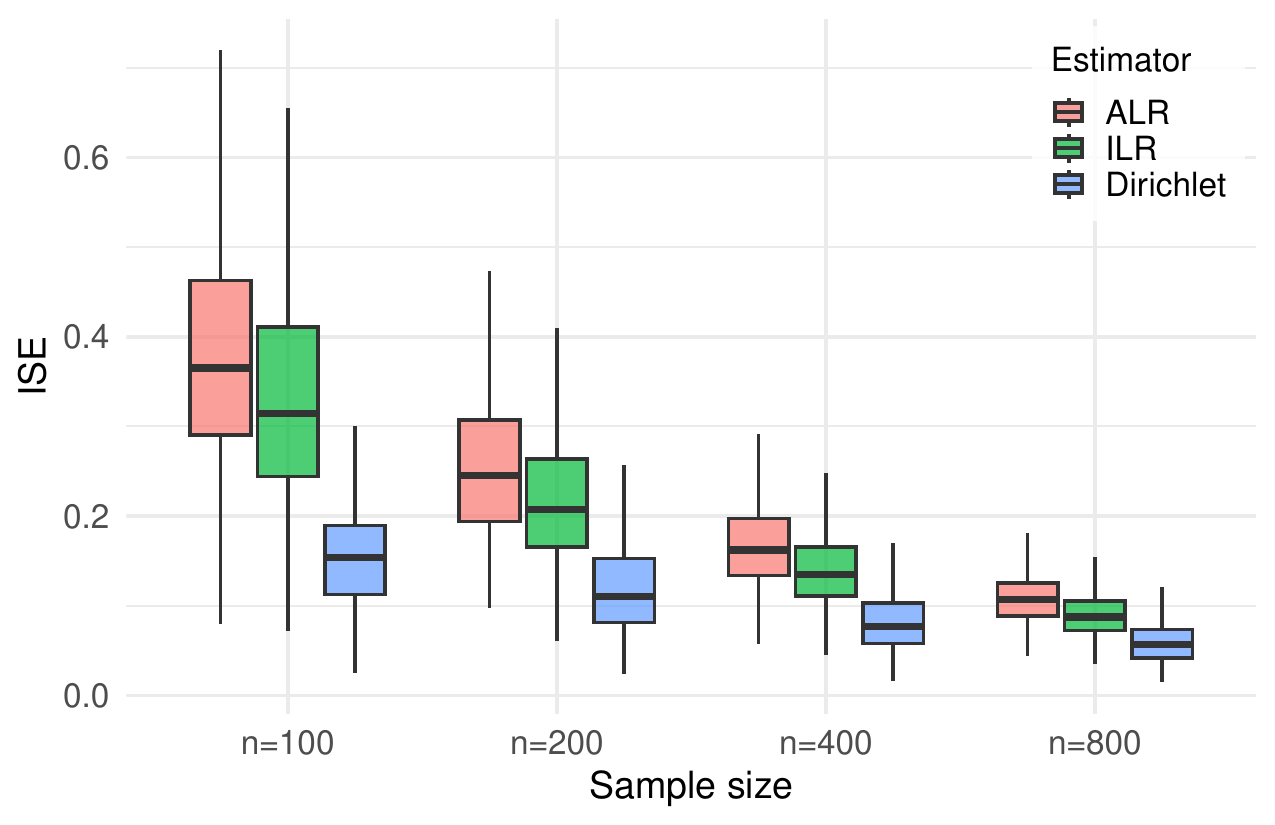}
\hfill
\includegraphics[width=0.49\textwidth]{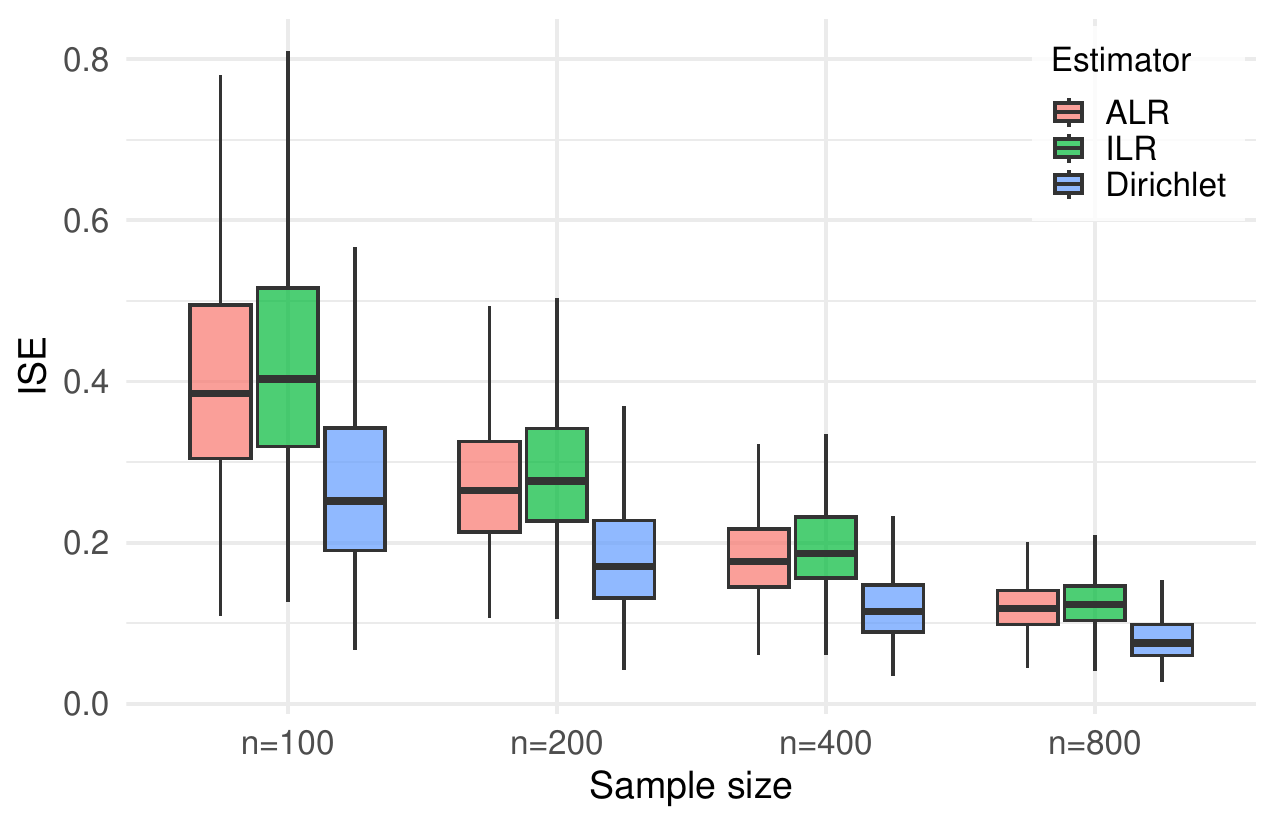}
\caption{Boxplots of $1000$ ISEs in Model~I (left panel) and Model~II (right panel) for alr, ilr and Dirichlet kernel estimators, under a MAR mechanism with a $20\%$ missing rate, and $n\in \{100, 200, 400, 800\}$.}
\label{fig:comparison.n}
\end{figure}

\subsubsection{Impact of the missing rate}

As the missing rate increases, inverse probability weights become more variable and the effective sample size of complete responses decreases, which inflates estimator variability. Consequently, the median and IQR of ISEs increase with the level of missingness for all estimators.

Figure~\ref{fig:comparison.missing} displays the boxplots of $1000$ ISEs at a fixed sample size $n=400$ across four missing rates: $5\%$, $10\%$, $20\%$, and $40\%$. For all missing rates and models, the IPW Dirichlet KDE performs best.

\begin{figure}[H]
\includegraphics[width=0.49\textwidth]{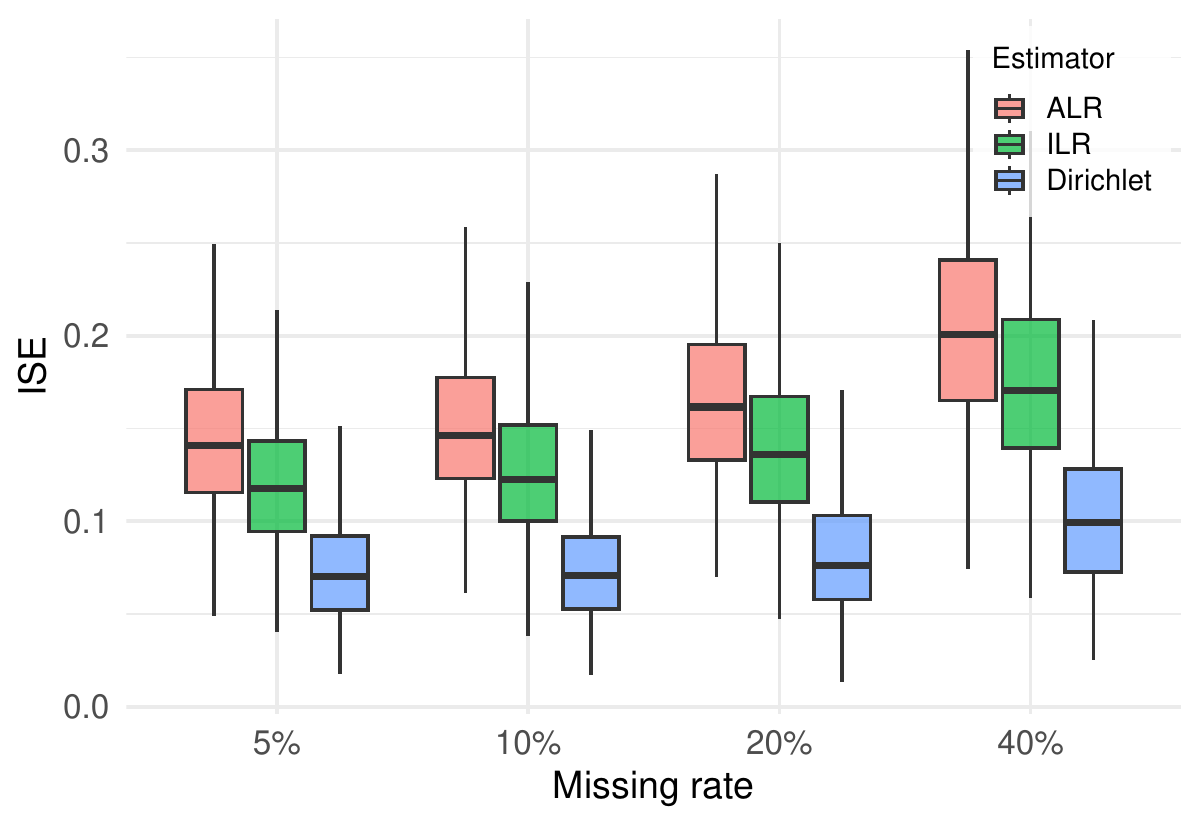}
\hfill
\includegraphics[width=0.49\textwidth]{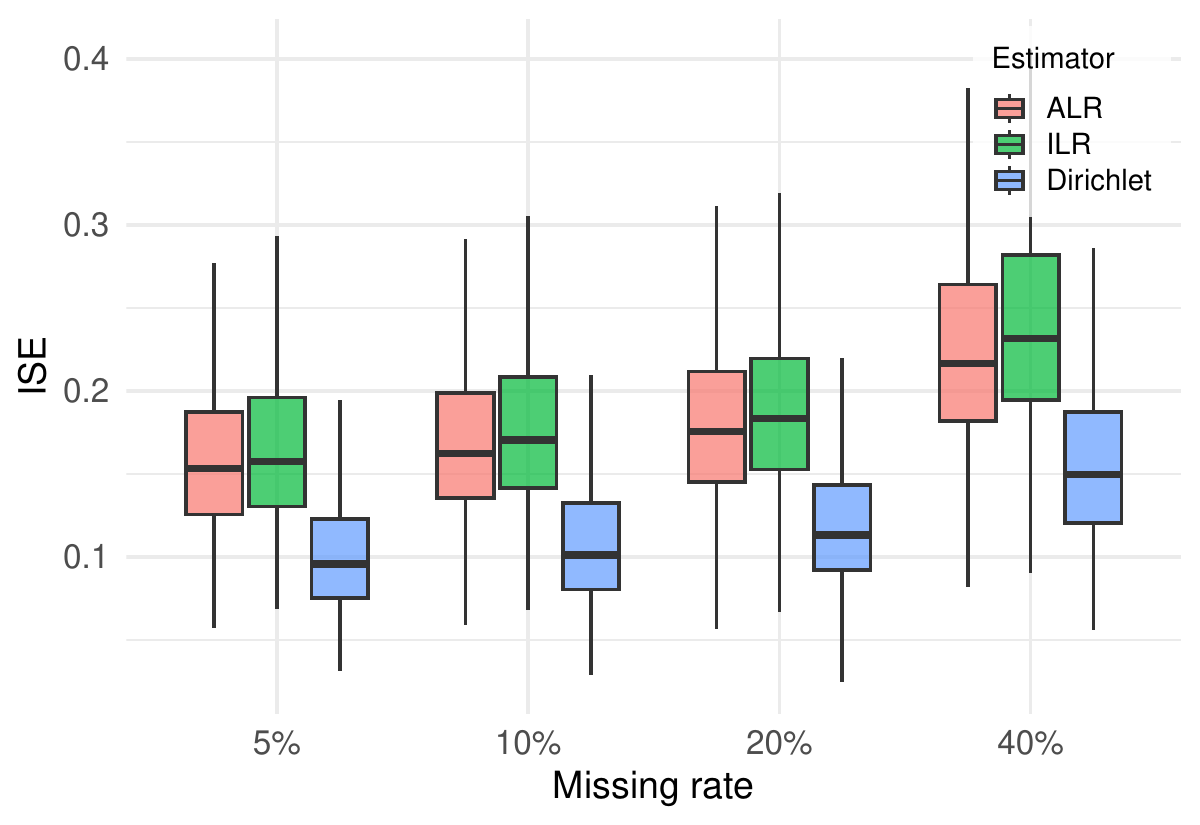}
\caption{Boxplots of $1000$ ISEs in Model~I (left panel) and Model~II (right panel) for alr, ilr and Dirichlet kernel estimators, under MAR mechanisms with $5\%$, $10\%$, $20\%$, and $40\%$ missing rates, and $n=400$.}
\label{fig:comparison.missing}
\end{figure}

\subsubsection{Joint effect of sample size and missing rate}

The joint variation of $n$ and the missing rate highlights a compensatory effect: while higher missingness increases variance through more variable IPW weights, larger sample sizes mitigate this inflation. In particular, at larger $n$ the dispersion of the ISE distributions decreases substantially, and performance remains stable even under $40\%$ missingness.

This interaction is summarized in Figure~\ref{fig:comparison.missing.n}, which compares ISE distributions across the sample sizes $n\in\{100,200,400,800\}$ for each missing rate ($10\%$, $20\%$, and $40\%$). For all sample sizes, missing rates, and models, the IPW Dirichlet KDE performs best.

\begin{figure}[H]
\includegraphics[width=0.49\textwidth]{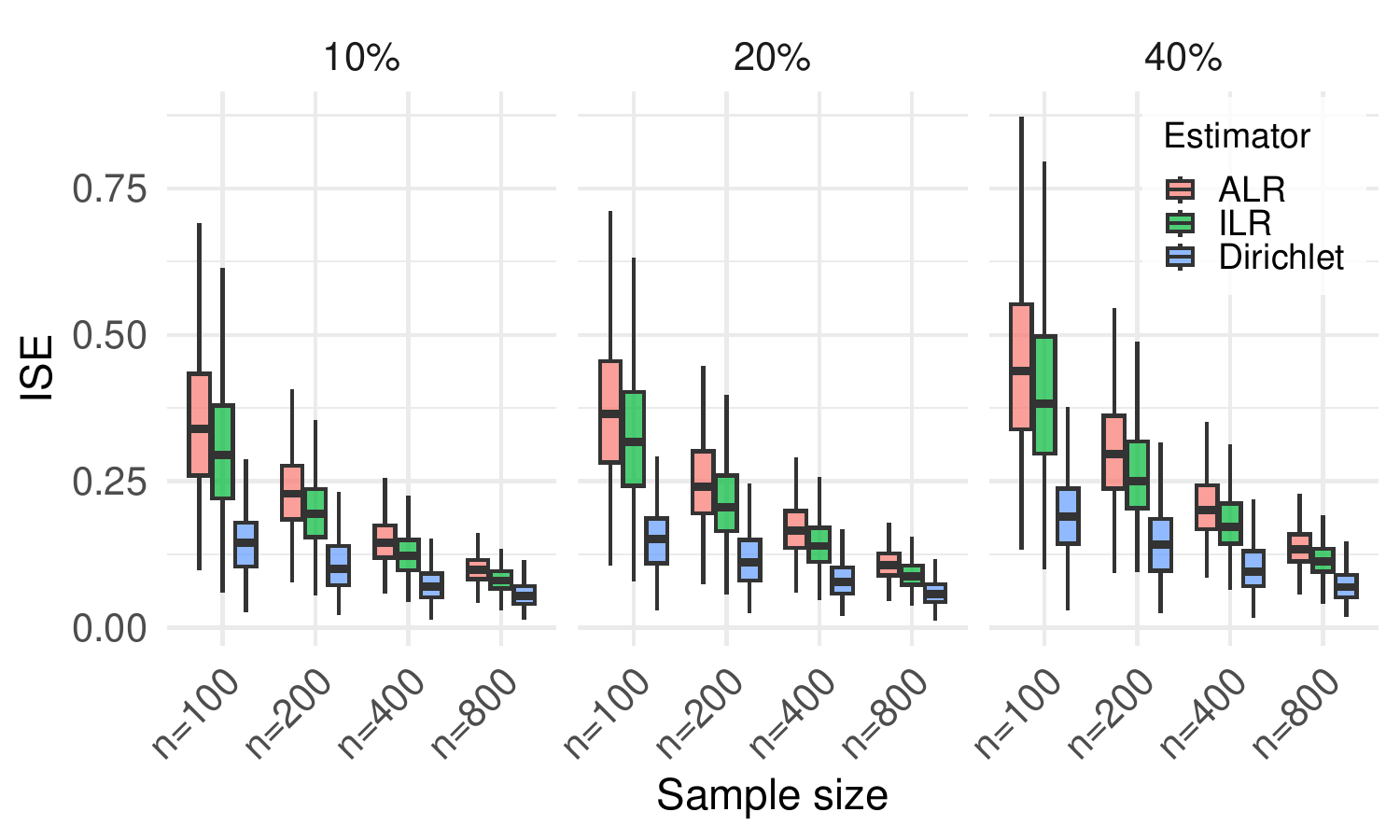}
\hfill
\includegraphics[width=0.49\textwidth]{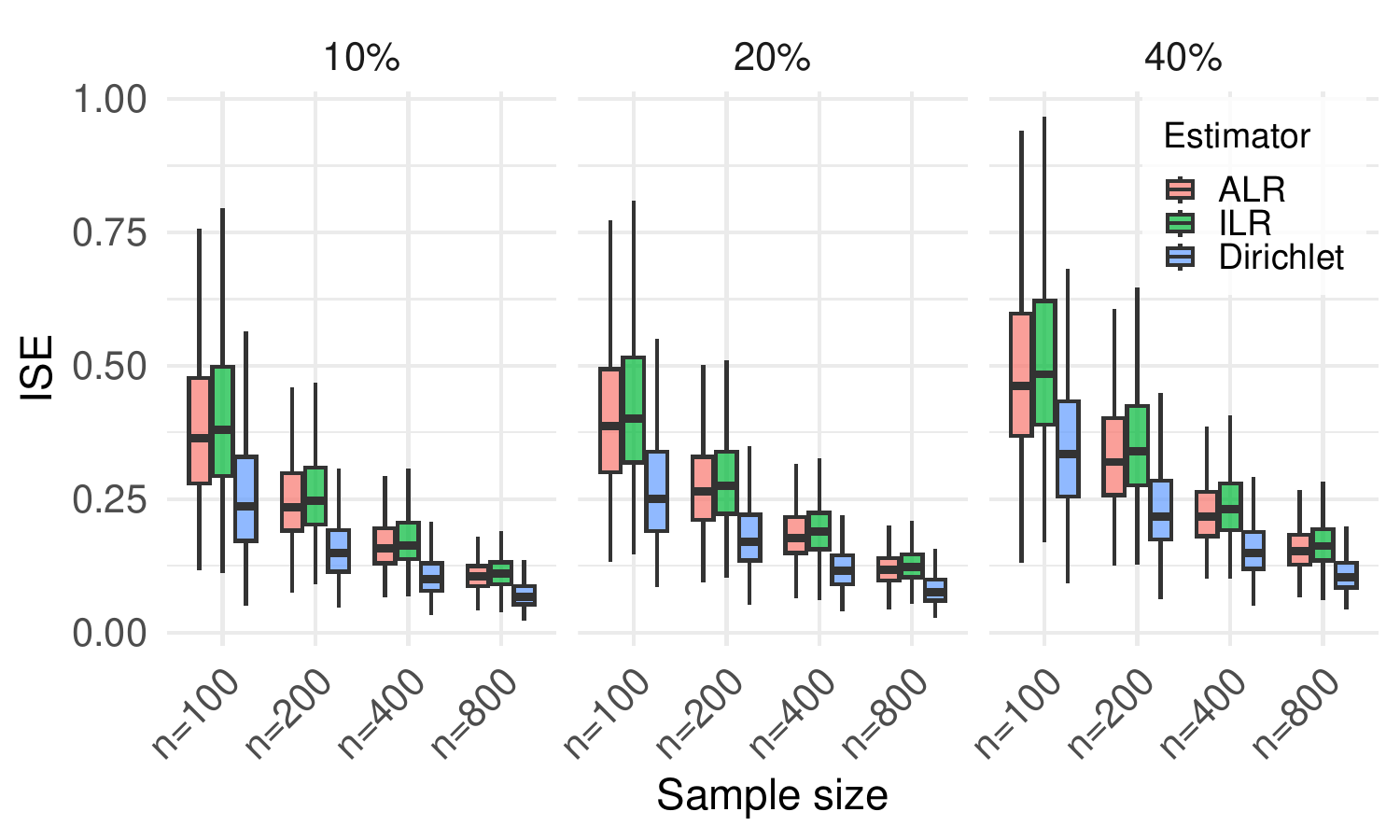}
\caption{Boxplots of $1000$ ISEs in Model~I (left panel) and Model~II (right panel) for alr, ilr and Dirichlet kernel estimators, under MAR mechanisms with 10\%, 20\%, and 40\% missing rates, and $n\in \{100, 200, 400, 800\}$.}
\label{fig:comparison.missing.n}
\end{figure}

Overall, the simulation study supports the practical usefulness of the proposed IPW Dirichlet KDE: it remains stable under moderate missingness, improves systematically with sample size, and compares favorably to log-ratio-based smoothing strategies in all settings examined here.

\section{Real-data application}\label{sec:application}

We illustrate the proposed methodology using data from the National Health and Nutrition Examination Survey (NHANES 2017--2018) conducted by the Centers for Disease Control and Prevention; see \citet{doi:10.32614/CRAN.package.nhanesA} to retrieve the dataset. NHANES is a large cross-sectional study designed to assess the health and nutritional status of the U.S. population through interviews and clinical examinations. In the present section, we use NHANES as a real-world illustration of the estimator; our analysis treats the extracted sample as iid and does not incorporate the complex survey design or examination weights, so results should be interpreted as descriptive of the analyzed NHANES sample rather than as fully population-representative.

We focus on the complete blood count (CBC) differential laboratory data (file \texttt{CBC\_J}), which provide the white blood cell differential (a key hematological indicator used in the diagnosis of infections, inflammatory conditions, and hematological disorders). The white blood cell differential is naturally compositional, as it represents relative proportions of leukocyte subtypes. For a three-part representation on the simplex, we consider the following complementary components:
\begin{itemize}[itemsep=0pt]
\item \textbf{Neutrophils} (\texttt{LBXNEPCT});
\item \textbf{Lymphocytes} (\texttt{LBXLYPCT});
\item \textbf{Others} (\texttt{LBXMOPCT} + \texttt{LBXEOPCT} + \texttt{LBXBAPCT}), defined as the aggregated proportion of Monocytes, Eosinophils, and Basophils.
\end{itemize}
In NHANES, these variables are reported as \emph{percentages} (suffix \texttt{PCT}, typically ranging from 0 to 100); throughout, we convert them to proportions by dividing by 100. Because laboratory percentages are subject to rounding, we additionally renormalize observed compositions so that the three parts sum exactly to one before applying the IPW Dirichlet KDE.

Each response, when observed, is a three-part composition. We represent it by $\bb{Y}_i=(Y_{i1}, Y_{i2}) \in \mathcal{S}_2$, with $Y_{i3}=1-\|\bb{Y}_i\|_1$. The three components correspond to the proportions of Neutrophils, Lymphocytes and Others in each leukocyte sample, respectively. In NHANES, the CBC differential is produced by a single laboratory procedure, so the differential is effectively observed \emph{as a block}: when the differential is unavailable or invalid, all its components are missing simultaneously for that participant. We therefore define the missingness indicator
\[
\delta_i = \ind_{\{\text{CBC differential fully observed}\}},
\]
i.e., $\delta_i=1$ when all required CBC differential percentages are available, and $\delta_i=0$ otherwise.

For the application below, we restrict attention to participants for whom body mass index (BMI) is observed. Specifically, after merging \texttt{CBC\_J} with the NHANES body measures file \texttt{BMX\_J} via the respondent identifier \texttt{SEQN} and retaining only records with observed \texttt{BMXBMI}, the analysis sample contains $n=8005$ participants. Among these $8005$ participants, the leukocyte composition $\bb{Y}_i$ is fully observed for $7280$ and missing for $725$, so the missing-response rate within the BMI-observed analysis sample is $725/8005 \approx 0.0906$ (about $9.1\%$). Thus, throughout this application, the auxiliary covariate is observed for every retained unit, whereas the response composition may still be missing.

In this framework, missingness of the response composition $\bb{Y}_i$ is assumed to be MAR, conditional on the observed covariate
\[
X_i = \texttt{BMXBMI}_i\in \R, \quad i\in [n].
\]
With these choices, the auxiliary dimension is $p=1$ and the composition lies in $\mathcal{S}_2$, so $d=2$. The observation probability $\pi(x)=\PP(\delta=1\mid X=x)$ is estimated nonparametrically using the Nadaraya--Watson estimator in \eqref{eq:NW}. To estimate the density of $\bb{Y}$, we use the feasible IPW Dirichlet KDE in \eqref{eq:feasible}. The Dirichlet bandwidth $b$ is selected by minimizing the IPW-adapted LSCV criterion introduced in \eqref{eq:LSCV_b_simulation}, yielding the data-driven choice $b^{\star}$ defined in \eqref{eq:b.star}.

The two- and three-dimensional representations in Figure~\ref{fig:application} reveal a dominant and biologically coherent mode of leukocyte composition, supporting the practical usefulness of the proposed IPW Dirichlet KDE on the simplex under missing-at-random sampling.

\begin{figure}[ht!]
\centering
\includegraphics[width=0.49\textwidth]{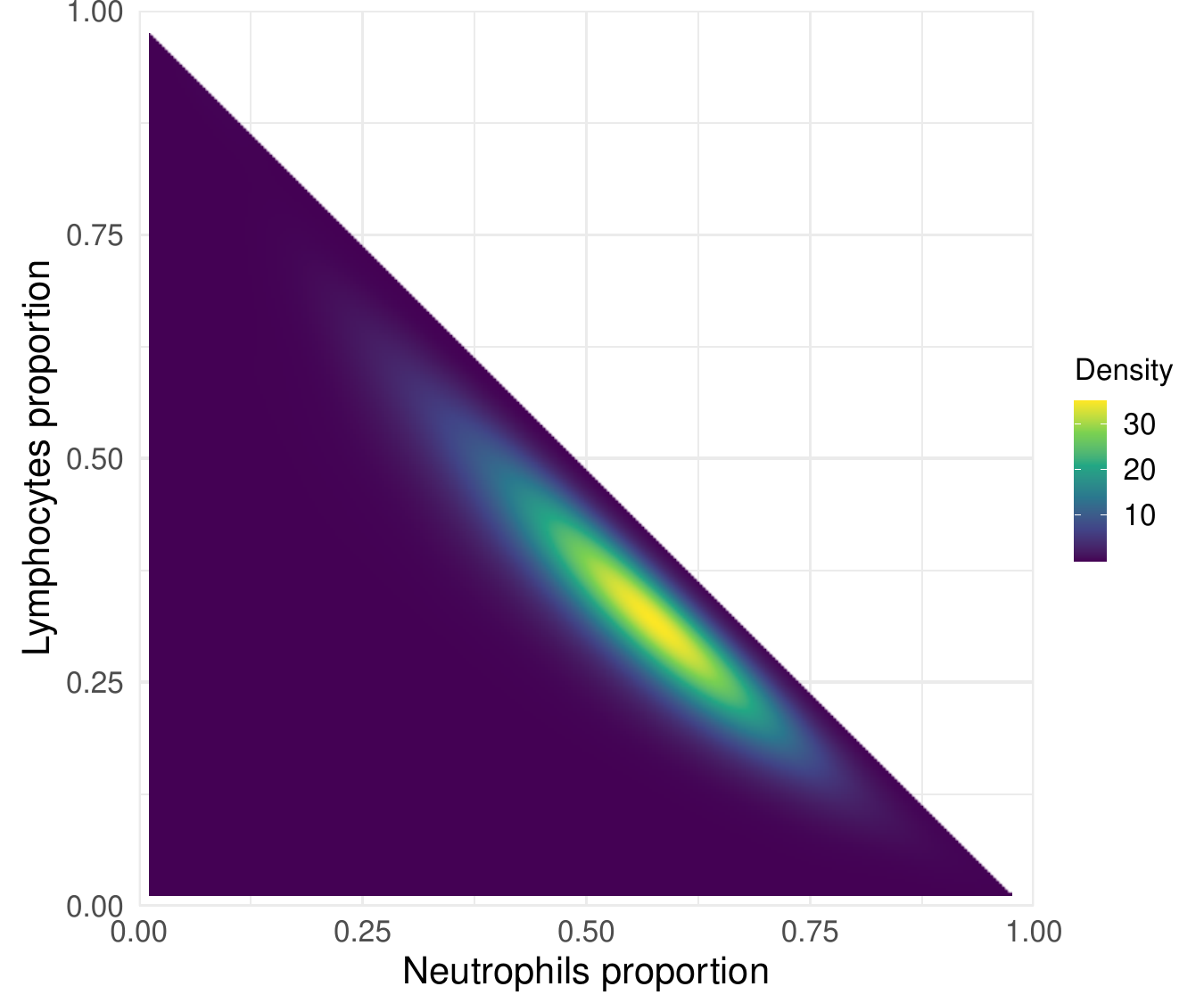}
\hfill
\includegraphics[trim=2.2cm 1.3cm 2.2cm 1.7cm, clip, width=0.49\textwidth]{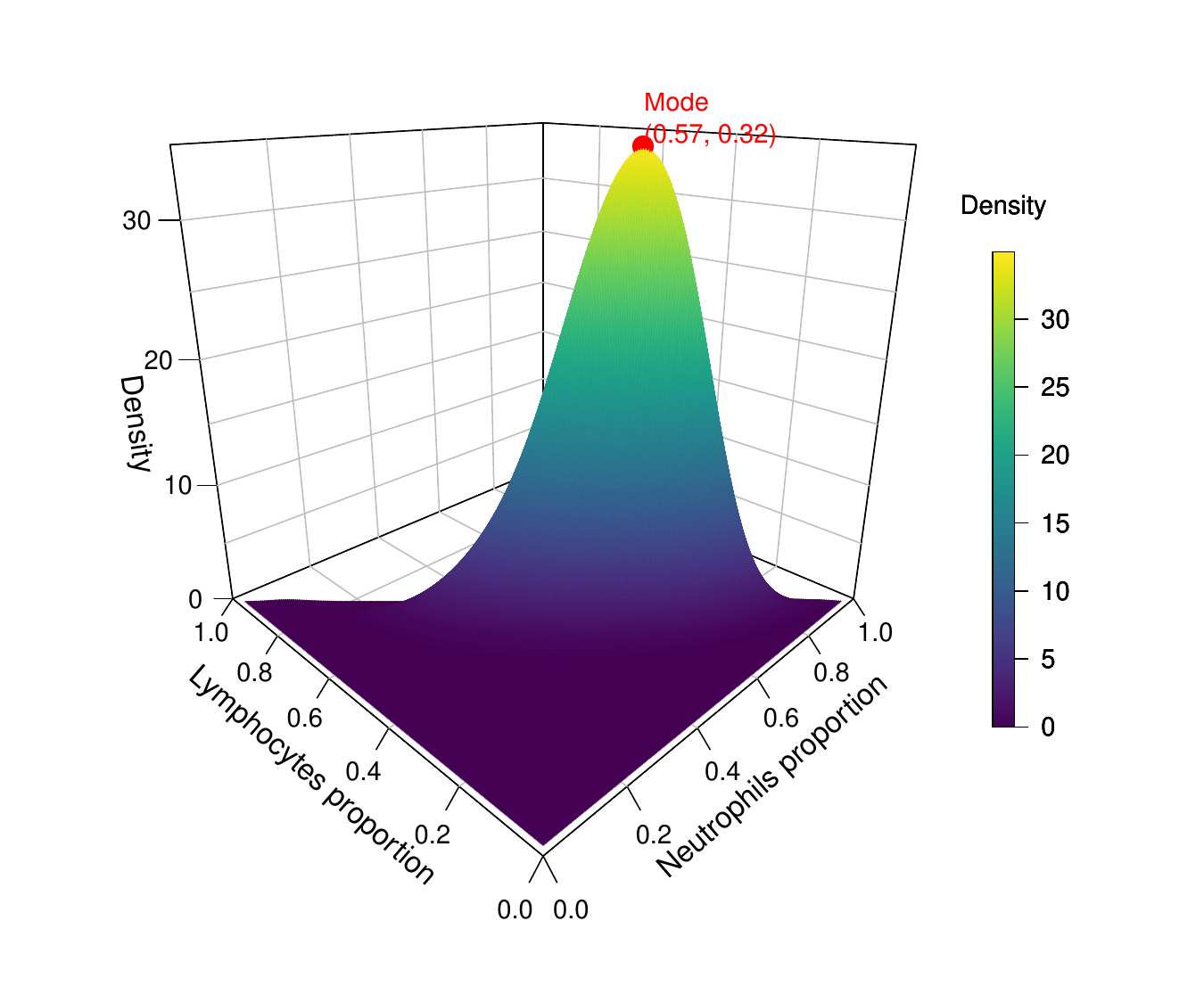}
\caption{Representations of the IPW Dirichlet kernel density estimate of leukocyte compositions.}
\label{fig:application}
\end{figure}

The estimated density attains its maximum in the vicinity of the composition $(0.57,\,0.32,\,0.11)$, corresponding to approximately $57\%$ Neutrophils, $32\%$ Lymphocytes, and $11\%$ Others. As illustrated in the right panel of Figure~\ref{fig:application}, this point represents the mode of the estimated density and can be interpreted as the most typical leukocyte profile, marking the region where observations from the analyzed NHANES sample are most highly concentrated.

From a biomedical standpoint, this modal composition is well within the established reference ranges for healthy adult populations; see, e.g., \citet{kratz2004laboratory} and \citet{Forget2017NLR}. The predominance of Neutrophils combined with a substantial proportion of Lymphocytes reflects a stable and balanced immune profile, often described as a state of immunological homeostasis \citep{buonacera2022neutrophil}.

Overall, this application illustrates how IPW Dirichlet smoothing provides a valid density estimate supported on the simplex while accounting for compositional constraints and blockwise missing CBC differentials.

\section{Summary and outlook}\label{sec:outlook}

This paper studied nonparametric density estimation for compositional data supported on the simplex under a MAR mechanism. Rather than imputing missing compositions, we proposed an IPW Dirichlet KDE that preserves the simplex geometry: it is nonnegative by construction, automatically respects the constrained support, and exhibits stable behavior near the boundary through the adaptive (asymmetric) Dirichlet kernel. We analyzed two variants: a pseudo estimator $\widetilde{f}_{n,b}$ with known propensity scores (Section~\ref{sec:results.pseudo}) and a feasible estimator $\hat{f}_{n,b}$ that replaces $\pi(\bb{X}_i)$ by a Nadaraya--Watson estimate $\hat{\pi}_i(\bb{X}_{1:n})$ (Section~\ref{sec:results.feasible}). For both estimators, we derived pointwise bias and variance expansions (Propositions~\ref{prop:bias.pseudo}--\ref{prop:variance.feasible}), MSE (Corollaries~\ref{cor:MSE.pseudo}--\ref{cor:MSE.feasible}), optimal smoothing rates, and asymptotic normality (Theorems~\ref{thm:CLT.pseudo}--\ref{thm:CLT.feasible}). A key message is that the leading bias term is the same as in the full-data Dirichlet KDE (Proposition~\ref{prop:bias.pseudo}), while MAR affects the variance through the additional factor $1+\zeta(\bb{s})$ in \eqref{eq:psi.zeta}. Moreover, when $\pi$ is estimated, the feasible estimator admits a second-order variance reduction term $-n^{-1}\xi(\bb{s})$ (Proposition~\ref{prop:variance.feasible}), so estimating the propensities need not inflate variability at the first order.

From a practical standpoint, the workflow is simple. Estimate the propensity scores $\hat{\pi}_i(\bb{X}_{1:n})$ using \eqref{eq:NW} (with a kernel bandwidth $h$ chosen, for instance, by a rule of thumb), and control extreme weights by enforcing a small floor $\pi_{\min}$ or using stabilized weights when needed. When auxiliary information includes discrete factors (e.g., sex), a simple strategy that preserves the theoretical setting is to stratify on the discrete variables and estimate $\pi$ within strata using only continuous covariates, or to extend \eqref{eq:NW} with mixed discrete/continuous kernels. Then select the Dirichlet bandwidth $b$ by minimizing the IPW-adapted LSCV criterion in \eqref{eq:LSCV_b_simulation} over a candidate grid, approximating the integral over $\mathcal{S}_2$ numerically on an interior grid. Our simulation results and the NHANES application illustrate that the resulting estimator is accurate and stable under moderate missingness, and that Dirichlet smoothing outperforms log-ratio-based alternatives on the simplex for the models we studied.

Several extensions are natural. First, sharper theory and data-driven tuning for the \emph{joint} choice of $(b,h)$ (e.g., joint cross-validation or cross-fitting for $\hat{\pi}$) could improve finite-sample performance. Second, accommodating sparse compositional data with structural zeros (common in microbiome applications) would require coupling the IPW Dirichlet KDE with principled zero-handling or zero-inflated models. Third, inference beyond pointwise limits, such as uniform confidence bands on $\mathcal{S}_d$, inference for functionals (e.g., modes, level sets, or integrated risks), and uncertainty quantification for the NHANES modal composition, remain important directions. Fourth, extending the methodology to complex survey designs (such as NHANES) by incorporating sampling weights and accounting for clustering/stratification would broaden the scope of valid population-level interpretations. Fifth, while the adaptivity of the Dirichlet kernel naturally mitigates boundary effects, developing higher-order bias reduction techniques, such as the non-negative multiplicative bias corrections adapted for multivariate bounded domains \citep{FunkeKawka2015}, within the IPW framework could further improve estimation accuracy near the simplex boundaries. Finally, it would be valuable to extend the approach to dependent designs (time series or spatial data) and to sensitivity analyses under departures from MAR (missing not at random), where IPW alone is no longer sufficient.

\section{Proofs}\label{sec:proofs}

\subsection{Proof of Proposition~\ref{prop:bias.pseudo}}

Let $\bb{s}\in \mathcal{S}_d$ be given. By linearity of expectations, the conditional independence of $\delta_i$ and $\bb{Y}_i$ given $\bb{X}_i$, and the fact that $\EE[\delta_1 \mid \bb{X}_1] = \pi(\bb{X}_1)$, we have
\[
\begin{aligned}
\EE[\widetilde{f}_{n,b}(\bb{s})]
&= n^{-1} \sum_{i=1}^n \EE\left[\frac{\delta_{i}}{\pi(\bb{X}_i)}\kappa_{\bb{s},b}(\bb{Y}_i)\right] \\
&= \EE\left[\EE\left[\frac{\delta_1}{\pi(\bb{X}_1)}\kappa_{\bb{s},b}(\bb{Y}_1) \mid \bb{X}_1\right]\right] \\
&= \EE\left[\frac{\EE[\delta_1 \mid \bb{X}_1]}{\pi(\bb{X}_1)} \EE[\kappa_{\bb{s},b}(\bb{Y}_1) \mid \bb{X}_1]\right] \\[2mm]
&= \EE[\EE[\kappa_{\bb{s},b}(\bb{Y}_1) \mid \bb{X}_1]] \\[3mm]
&= \EE[\kappa_{\bb{s},b}(\bb{Y}_1)],
\end{aligned}
\]
so that
\[
\EE[\widetilde{f}_{n,b}(\bb{s})]
= \EE[\kappa_{\bb{s},b}(\bb{Y}_1)]
= \EE[\hat{f}^{\hspace{0.2mm}\mathrm{full}}_{n,b}(\bb{s})].
\]
From here, the asymptotics are exactly the same as in Theorem~1 of \citet{MR4319409} under Assumptions~\ref{ass:3} and \ref{ass:1}. This concludes the proof.

\subsection{Proof of Proposition~\ref{prop:variance.pseudo}}

Let $\bb{s}\in \mathrm{Int}(\mathcal{S}_d)$ be given such that $f(\bb{s})\in (0,\infty)$. Since the pairs $(\bb{X}_1,\bb{Y}_1),\ldots,(\bb{X}_n,\bb{Y}_n)$ are iid, note that
\[
\begin{aligned}
\Var(\widetilde{f}_{n,b}(\bb{s}))
&= \Var\left(n^{-1} \sum_{i=1}^n \frac{\delta_{i}}{\pi(\bb{X}_i)} \kappa_{\bb{s},b}(\bb{Y}_i)\right) \\
&= n^{-1} \Var\left(\frac{\delta_1}{\pi(\bb{X}_1)} \kappa_{\bb{s},b}(\bb{Y}_1)\right).
\end{aligned}
\]
By the law of total variance, we have the decomposition:
\[
\begin{aligned}
\Var\left(\frac{\delta_1}{\pi(\bb{X}_1)} \kappa_{\bb{s},b}(\bb{Y}_1)\right)
&= \Var\left(\EE\left[\frac{\delta_1}{\pi(\bb{X}_1)} \kappa_{\bb{s},b}(\bb{Y}_1) \mid \bb{X}_1\right]\right) + \EE\left[\Var\left(\frac{\delta_1}{\pi(\bb{X}_1)} \kappa_{\bb{s},b}(\bb{Y}_1)\mid \bb{X}_1\right)\right] \\
&\equiv I_1 + I_2.
\end{aligned}
\]
It follows from the calculations in the proof of Proposition~\ref{prop:bias.pseudo} that
\[
I_1 = \Var\left(\EE\left[\kappa_{\bb{s},b}(\bb{Y}_1) \mid \bb{X}_1\right]\right).
\]
Applying the law of total variance, together with $\Var(\delta_1\mid \bb{Y}_1,\bb{X}_1) = \pi(\bb{X}_1) (1 - \pi(\bb{X}_1))$, yields
\[
\begin{aligned}
&\Var\left(\frac{\delta_1}{\pi(\bb{X}_1)} \kappa_{\bb{s},b}(\bb{Y}_1)\mid \bb{X}_1\right) \\
&\quad= \Var\left(\EE\left[\frac{\delta_1}{\pi(\bb{X}_1)} \kappa_{\bb{s},b}(\bb{Y}_1) \mid \bb{Y}_1, \bb{X}_1\right] \mid \bb{X}_1\right) + \EE\left[\Var\left(\frac{\delta_1}{\pi(\bb{X}_1)} \kappa_{\bb{s},b}(\bb{Y}_1)\mid \bb{Y}_1, \bb{X}_1\right) \mid \bb{X}_1\right] \\
&\quad= \Var\left(\frac{\EE[\delta_1 \mid \bb{Y}_1,\bb{X}_1]}{\pi(\bb{X}_1)} \kappa_{\bb{s},b}(\bb{Y}_1) \mid \bb{X}_1\right) + \EE\left[\frac{\Var(\delta_1\mid \bb{Y}_1,\bb{X}_1)}{\pi(\bb{X}_1)^2} \kappa_{\bb{s},b}(\bb{Y}_1)^2 \mid \bb{X}_1\right] \\
&\quad= \Var(\kappa_{\bb{s},b}(\bb{Y}_1) \mid \bb{X}_1) + \frac{1 - \pi(\bb{X}_1)}{\pi(\bb{X}_1)} \EE\left[\kappa^2_{\bb{s},b}(\bb{Y}_1) \mid \bb{X}_1\right],
\end{aligned}
\]
so that
\[
I_2 = \EE[\Var(\kappa_{\bb{s},b}(\bb{Y}_1) \mid \bb{X}_1)] + \EE\left[\frac{1 - \pi(\bb{X}_1)}{\pi(\bb{X}_1)} \EE[\kappa^2_{\bb{s},b}(\bb{Y}_1) \mid \bb{X}_1]\right].
\]
Therefore, a third application of the law of total variance implies
\begin{equation}\label{eq:var.end.0}
\begin{aligned}
\Var(\widetilde{f}_{n,b}(\bb{s}))
&= n^{-1} (I_1 + I_2) \\
&= n^{-1} \Var(\kappa_{\bb{s},b}(\bb{Y}_1)) + n^{-1} \EE\left[\frac{1 - \pi(\bb{X}_1)}{\pi(\bb{X}_1)} \EE[\kappa^2_{\bb{s},b}(\bb{Y}_1) \mid \bb{X}_1]\right].
\end{aligned}
\end{equation}
By Theorem~2 of \citet{MR4319409} under Assumption~\ref{ass:2} we know that
\begin{equation}\label{eq:var.end.1}
n^{-1} \Var(\kappa_{\bb{s},b}(\bb{Y}_1)) = \Var(\hat{f}^{\hspace{0.2mm}\mathrm{full}}_{n,b}(\bb{s})) = n^{-1} b^{-d/2} \psi(\bb{s}) f(\bb{s}) (1 + \OO_{\bb{s}}(b^{1/2})).
\end{equation}
Similarly, by applying Lemma~1 of \citet{MR4319409} under Assumption~\ref{ass:6}, with $f_{\bb{Y}_1 \mid \bb{X}_1}$ as the target density, we have, for each $\bb{x}\in \R^p$ such that $g(\bb{x})\in (0,\infty)$,
\[
\EE[\kappa^2_{\bb{s},b}(\bb{Y}_1) \mid \bb{X}_1 = \bb{x}] = b^{-d/2} \psi(\bb{s}) f_{\bb{Y}_1 \mid \bb{X}_1}(\bb{s} \mid \bb{x}) (1 + \OO_{\bb{s}}(b^{1/2})),
\]
Hence, since $f_{\bb{Y}_1 \mid \bb{X}_1}(\bb{s} \mid \bb{x}) g(\bb{x}) / f(\bb{s}) = f_{\bb{X}_1 \mid \bb{Y}_1}(\bb{x} \mid \bb{s})$ by Bayes' rule, we obtain
\begin{equation}\label{eq:var.end.2}
\begin{aligned}
&n^{-1} \EE\left[\frac{1 - \pi(\bb{X}_1)}{\pi(\bb{X}_1)} \EE[\kappa^2_{\bb{s},b}(\bb{Y}_1) \mid \bb{X}_1]\right] \\
&\qquad= n^{-1} b^{-d/2} \psi(\bb{s}) f(\bb{s}) \, \EE\left[\frac{1 - \pi(\bb{X}_1)}{\pi(\bb{X}_1)} \mid \bb{Y}_1 = \bb{s}\right] (1 + \OO_{\bb{s}}(b^{1/2})).
\end{aligned}
\end{equation}
(Note that the last conditional expectation exists by Assumption~\ref{ass:5}.) Substituting \eqref{eq:var.end.1} and \eqref{eq:var.end.2} back into \eqref{eq:var.end.0}, the conclusion follows.

\subsection{Proof of Theorem~\ref{thm:CLT.pseudo}}

Let $\bb{s}\in \mathrm{Int}(\mathcal{S}_d)$ be given such that $f(\bb{s})\in (0,\infty)$. Choose $b\sim c n^{-2/(d+4)}$ as $n\to \infty$ for some $c\in (0,\infty)$. Start with the decomposition:
\[
\begin{aligned}
&n^{1/2} b^{\hspace{0.2mm}d/4} \big(\widetilde{f}_{n,b}(\bb{s}) - f(\bb{s}) - b \phi(\bb{s})\big) \\
&\quad= n^{1/2} b^{\hspace{0.2mm}d/4} \big(\widetilde{f}_{n,b}(\bb{s}) - \EE[\widetilde{f}_{n,b}(\bb{s})]\big) + n^{1/2} b^{\hspace{0.2mm}d/4} \big(\EE[\widetilde{f}_{n,b}(\bb{s})] - f(\bb{s}) - b \phi(\bb{s})\big).
\end{aligned}
\]
The second term on the right-hand side is $\oo(n^{1/2} b^{(d+4)/4})$ by Proposition~\ref{prop:bias.pseudo} (under Assumption~\ref{ass:3}), which is $\oo(1)$ because $b\sim c n^{-2/(d+4)}$. Thus, it suffices to show that
\[
n^{1/2} b^{\hspace{0.2mm}d/4} (\widetilde{f}_{n,b}(\bb{s}) - \EE[\widetilde{f}_{n,b}(\bb{s})])\rightsquigarrow \mathcal{N}(0, \psi(\bb{s}) f(\bb{s}) (1 + \zeta(\bb{s}))).
\]
Note that we can write
\[
\widetilde{f}_{n,b}(\bb{s}) - \EE[\widetilde{f}_{n,b}(\bb{s})] = n^{-1} \sum_{i=1}^n Z_{i,b}(\bb{s}),
\]
where the random variables,
\[
Z_{i,b}(\bb{s}) = \frac{\delta_{i}}{\pi(\bb{X}_i)} \kappa_{\bb{s},b}(\bb{Y}_i)-\EE\left[\frac{\delta_{i}}{\pi(\bb{X}_i)} \kappa_{\bb{s},b}(\bb{Y}_i)\right], \quad i\in \{1,\ldots,n\},
\]
are centered and independent. The asymptotic normality of $n^{1/2} b^{\hspace{0.2mm}d/4} (\widetilde{f}_{n,b}(\bb{s}) - \EE \{\widetilde{f}_{n,b}(\bb{s})\})$ will be proved if we verify the following Lindeberg condition for double arrays: For every $\e\in (0,\infty)$,
\begin{equation}\label{eq:Lindeberg.condition}
s_{n,b}^{-2} \, \EE[|Z_{1,b}(\bb{s})|^2 \, \mathds{1}_{\{|Z_{1,b}(\bb{s})| > \e n^{1/2} s_{n,b}\}}] \to 0, \quad n\to \infty,
\end{equation}
where $s_{n,b}^2 = \EE[|Z_{1,b}(\bb{s})|^2]$. Using Assumption~\ref{ass:5}, and applying the local bound on the Dirichlet kernel from Lemma~\ref{lem:local.bound}, one finds that, for every $i\in \{1,\ldots,n\}$,
\[
|Z_{i,b}(\bb{s})| \ll \frac{\max_{\bb{x}\in \mathcal{S}_d} \kappa_{\bb{s},b}(\bb{x})}{\pi(\bb{X}_i)} \ll_{\bb{s}} \frac{b^{-d/2}}{\pi_{\min}} = \OO_{\bb{s}}(b^{-d/2}).
\]
Using Proposition~\ref{prop:variance.pseudo} (under Assumptions~\ref{ass:3} and \ref{ass:6}), we have
\[
s_{n,b}^2 = b^{-d/2} \psi(\bb{s}) f(\bb{s}) (1 + \zeta(\bb{s})) (1 + \oo_{\bb{s}}(1)).
\]
Combining the last two equations yields
\[
\frac{|Z_{1,b}(\bb{s})|}{n^{1/2} s_{n,b}} \ll_{\bb{s}} n^{-1/2} \, b^{-d/2}\, b^{\hspace{0.2mm}d/4} = n^{-1/2} \, b^{-d/4} \to 0.
\]
It follows that the indicator function in \eqref{eq:Lindeberg.condition} is eventually $0$ (uniformly in $\omega$), which proves Lindeberg's condition. This concludes the proof.

\subsection{Proof of Proposition~\ref{prop:bias.feasible}}

Let $\bb{s}\in \mathcal{S}_d$ be given. The Taylor expansion of $1/\hat{\pi}_i(\bb{X}_{1:n})$ around $1/\pi(\bb{X}_i)$ yields
\[
\frac{1}{\hat{\pi}_i(\bb{X}_{1:n})} = \frac{1}{\pi(\bb{X}_i)} - \frac{1}{\pi(\bb{X}_i)^2} (\hat{\pi}_i(\bb{X}_{1:n}) - \pi(\bb{X}_i)) + \sum_{j=2}^{\infty} \frac{(-1)^j}{\pi(\bb{X}_i)^{j+1}} (\hat{\pi}_i(\bb{X}_{1:n}) - \pi(\bb{X}_i))^j.
\]
It follows that
\begin{equation}\label{eq:hat.f.expansion}
\begin{aligned}
\hat{f}_{n,b}(\bb{s})
&= n^{-1} \sum_{i=1}^n \frac{\delta_i}{\pi(\bb{X}_i)} \kappa_{\bb{s},b}(\bb{Y}_i) - n^{-1} \sum_{i=1}^n \frac{\delta_i}{\pi(\bb{X}_i)^2}(\hat{\pi}_i(\bb{X}_{1:n}) - \pi(\bb{X}_i)) \kappa_{\bb{s},b}(\bb{Y}_i) \\
&\qquad+ n^{-1} \sum_{j=2}^{\infty} \sum_{i=1}^n \frac{(-1)^j \delta_i}{\pi(\bb{X}_i)^{j+1}}(\hat{\pi}_i(\bb{X}_{1:n}) - \pi(\bb{X}_i))^j \kappa_{\bb{s},b}(\bb{Y}_i) \\
&\equiv \widetilde{f}_{n,b}(\bb{s}) - T_{n,1} + T_{n,2}.
\end{aligned}
\end{equation}

First, we control the term $T_{n,2}$ in \eqref{eq:hat.f.expansion}. Under Assumption~\ref{ass:5} and Assumptions~\ref{ass:B.1}--\ref{ass:B.3}, and our choice $h \equiv h_n \asymp n^{-1/(p+4)}$ as $n\to \infty$, we know from a multivariate extension of Corollaries~2.1--2.2 in \citet{MR3167772} that, for all $k\in \N$ and $\bb{x}\in \R^p$ such that $g(\bb{x})\in (0,\infty)$,
\begin{equation}\label{eq:Geenens}
\EE[|\hat{\pi}_i(\bb{X}_{1:n}) - \pi(\bb{x})|^k \mid \bb{X}_i = \bb{x}] \ll n^{-2k/(p+4)}.
\end{equation}
By the triangle inequality and the conditional independence of $\delta_i$ and $\bb{Y}_i$ given $\bb{X}_i$, we have
\[
\begin{aligned}
|\EE[T_{n,2}]|
&= \left|\EE\left[n^{-1} \sum_{j=2}^{\infty} \sum_{i=1}^n \frac{(-1)^j \delta_i}{\pi(\bb{X}_i)^{j+1}}(\hat{\pi}_i(\bb{X}_{1:n}) - \pi(\bb{X}_i))^j \kappa_{\bb{s},b}(\bb{Y}_i)\right]\right| \\
&\quad\leq n^{-1} \sum_{j=2}^{\infty} \frac{1}{\pi_{\min}^{j+1}} \sum_{i=1}^n \EE\left[\EE[|\hat{\pi}_i(\bb{X}_{1:n}) - \pi(\bb{X}_i)|^j \mid \bb{X}_i] \, \EE[\kappa_{\bb{s},b}(\bb{Y}_i) \mid \bb{X}_i]\right].
\end{aligned}
\]
Combined with \eqref{eq:Geenens}, the fact that the conditional density $f_{\bb{Y}_i \mid \bb{X}_i}$ is bounded under Assumption~\ref{ass:6}, and the fact that $\int_{\mathcal{S}_d} \kappa_{\bb{s},b}(\bb{y}) \rd \bb{y} = 1$, we deduce that
\begin{equation}\label{eq:bias.T3}
|\EE[T_{n,2}]|
\ll \sum_{j=2}^{\infty} \frac{1}{\pi_{\min}^{j+1}} n^{-2j/(p+4)}
\ll n^{-4/(p+4)}.
\end{equation}

Next, decompose the term $T_{n,1}$ on the right-hand side of \eqref{eq:hat.f.expansion} as follows:
\[
\begin{aligned}
T_{n,1}
&= n^{-1} \sum_{i=1}^n \frac{\delta_i - \pi(\bb{X}_i)}{\pi(\bb{X}_i)^2}(\hat{\pi}_i(\bb{X}_{1:n}) - \pi(\bb{X}_i)) \kappa_{\bb{s},b}(\bb{Y}_i) \\
&\quad+ n^{-1} \sum_{i=1}^n \frac{\pi(\bb{X}_i)}{\pi(\bb{X}_i)^2}(\hat{\pi}_i(\bb{X}_{1:n}) - \pi(\bb{X}_i)) \kappa_{\bb{s},b}(\bb{Y}_i) \\
&\equiv J_{n,1} + J_{n,2}.
\end{aligned}
\]
Letting
\begin{equation}\label{eq:g.hat}
\hat{g}_n(\bb{x}) = n^{-1} \sum_{j=1}^n K_h^*(\bb{x} - \bb{X}_j), \quad \bb{x}\in \R^p,
\end{equation}
we have
\begin{align}
J_{n,1}
&= n^{-1} \sum_{i=1}^n \frac{(\delta_i - \pi(\bb{X}_i))}{\pi(\bb{X}_i)^2} \left\{\frac{\sum_{j=1}^n \delta_j \, K_h^*(\bb{X}_i - \bb{X}_j)}{\sum_{j=1}^n K_h^*(\bb{X}_i - \bb{X}_j)} - \pi(\bb{X}_i)\right\} \kappa_{\bb{s},b}(\bb{Y}_i) \\
&= n^{-2} \sum_{i=1}^n \sum_{j=1}^n \frac{(\delta_i - \pi(\bb{X}_i)) (\delta_j - \pi(\bb{X}_i)) \, K_h^*(\bb{X}_i - \bb{X}_j) \, \kappa_{\bb{s},b}(\bb{Y}_i)}{\pi^2(\bb{X}_i) \, \hat{g}_n(\bb{X}_i)}
\end{align}
Conditioning on $\bb{X}_{1:n}$ yields, for $i\neq j$,
\begin{equation}\label{eq:cov.delta}
\begin{aligned}
\EE[(\delta_i - \pi(\bb{X}_i))(\delta_j - \pi(\bb{X}_i)) \mid \bb{X}_{1:n}]
&= \EE[\delta_i - \pi(\bb{X}_i) \mid \bb{X}_{1:n}] \, \EE[\delta_j - \pi(\bb{X}_i) \mid \bb{X}_{1:n}] \\
&= (\pi(\bb{X}_i) - \pi(\bb{X}_i)) (\pi(\bb{X}_j) - \pi(\bb{X}_i)) \\
&= 0,
\end{aligned}
\end{equation}
so that
\[
\begin{aligned}
\EE[J_{n,1}]
&= n^{-2} \sum_{i=1}^n \sum_{j=1}^n \EE\left[\frac{\EE[(\delta_i - \pi(\bb{X}_i)) (\delta_j - \pi(\bb{X}_i)) \mid \bb{X}_{1:n}] \, K_h^*(\bb{X}_i - \bb{X}_j) \, \EE[\kappa_{\bb{s},b}(\bb{Y}_i) \mid \bb{X}_{1:n}]}{\pi^2(\bb{X}_i) \, \hat{g}_n(\bb{X}_i)}\right] \\
&= n^{-2} \sum_{i=1}^n \EE\left[\frac{\EE[(\delta_i - \pi(\bb{X}_i))^2 \mid \bb{X}_{1:n}] \, K_h^*(\bb{0}_p) \, \EE[\kappa_{\bb{s},b}(\bb{Y}_i) \mid \bb{X}_{1:n}]}{\pi^2(\bb{X}_i) \, \hat{g}_n(\bb{X}_i)}\right] \\
&= K^*(\bb{0}_p) \, n^{-1} h^{-p} \, \EE\left[\frac{\EE[(\delta_1 - \pi(\bb{X}_1))^2 \mid \bb{X}_{1:n}] \, \EE[\kappa_{\bb{s},b}(\bb{Y}_1) \mid \bb{X}_{1:n}]}{\pi^2(\bb{X}_1) \, \hat{g}_n(\bb{X}_1)}\right].
\end{aligned}
\]
Since
\[
\EE[(\delta_1 - \pi(\bb{X}_1))^2 \mid \bb{X}_{1:n}] = \Var(\delta_1 \mid \bb{X}_{1:n}) = \pi(\bb{X}_1) (1 - \pi(\bb{X}_1)),
\]
we have, under Assumption~\ref{ass:6},
\[
\begin{aligned}
\EE[J_{n,1}]
&= K^*(\bb{0}_p) \, n^{-1} h^{-p} \, \EE\left[\frac{(1 - \pi(\bb{X}_1)) \, \EE[\kappa_{\bb{s},b}(\bb{Y}_1) \mid \bb{X}_{1:n}]}{\pi(\bb{X}_1) \, \hat{g}_n(\bb{X}_1)}\right] \\
&\ll n^{-1} h^{-p} \, \EE\left[\frac{1}{\hat{g}_n(\bb{X}_1)}\right].
\end{aligned}
\]
By Assumptions~\ref{ass:B.1}--\ref{ass:B.3}, the uniform strong consistency of $\hat{g}_n$, and Assumption~\ref{ass:4}, it follows that
\[
\EE[J_{n,1}] \ll n^{-1} h^{-p} \, \EE\left[\frac{1}{g(\bb{X}_1)}\right] \ll n^{-1} h^{-p}.
\]

Next, we want to control $J_{n,2}$. Under Assumptions~\ref{ass:B.1}--\ref{ass:B.3}, note that
\[
\EE[\hat{\pi}_i(\bb{X}_{1:n}) \mid \bb{X}_i] - \pi(\bb{X}_i) = \OO(h^2),
\]
being the bias term of a classical Nadaraya--Watson regression estimator. Therefore, under Assumptions~\ref{ass:5}--\ref{ass:6}, we have
\[
\begin{aligned}
|\EE[J_{n,2}]|
&= \left|\EE\left[n^{-1} \sum_{i=1}^n \frac{1}{\pi(\bb{X}_i)}(\EE[\hat{\pi}_i(\bb{X}_{1:n}) \mid \bb{X}_i] - \pi(\bb{X}_i)) \, \EE[\kappa_{\bb{s},b}(\bb{Y}_i) \mid \bb{X}_i]\right]\right| \\
&\leq \EE\left[n^{-1} \sum_{i=1}^n \frac{1}{\pi(\bb{X}_i)} |\EE[\hat{\pi}_i(\bb{X}_{1:n}) \mid \bb{X}_i] - \pi(\bb{X}_i)| \, \EE[\kappa_{\bb{s},b}(\bb{Y}_i) \mid \bb{X}_i]\right] \\
&\ll \frac{h^2}{\pi_{\mathrm{min}}} \, \EE\left[\kappa_{\bb{s},b}(\bb{Y}_1) \mid \bb{X}_1\right] \\
&\ll h^2.
\end{aligned}
\]
We deduce that
\begin{equation}\label{eq:bias.T2}
\EE[T_{n,1}] = \EE[J_{n,1}] + \EE[J_{n,2}] = \OO(n^{-1} h^{-p} + h^2).
\end{equation}

Substituting \eqref{eq:bias.T3} and \eqref{eq:bias.T2} back into \eqref{eq:hat.f.expansion}, we have, for $h\sim \kappa n^{-1/(p+4)}$,
\[
\EE(\hat{f}_{n,b}(\bb{s})) = \EE(\widetilde{f}_{n,b}(\bb{s})) + \OO(n^{-2/(p+4)}).
\]
The asymptotics of $\EE(\widetilde{f}_{n,b}(\bb{s}))$ are given in Proposition~\ref{prop:bias.pseudo}. This concludes the proof.

\subsection{Proof of Proposition~\ref{prop:variance.feasible}}

Let $\bb{s}\in \mathrm{Int}(\mathcal{S}_d)$ be given such that $f(\bb{s})\in (0,\infty)$. As in the proof of Proposition~\ref{prop:bias.feasible}, consider the decomposition:
\begin{equation}\label{eq:hat.f.expansion.2}
\hat{f}_{n,b}(\bb{s}) = \widetilde{f}_{n,b}(\bb{s}) + \sum_{j=1}^{\infty} (-1)^j M_{n,j},
\end{equation}
where
\[
M_{n,j} = n^{-1} \sum_{i=1}^n \frac{\delta_i}{\pi(\bb{X}_i)^{j+1}}(\hat{\pi}_i(\bb{X}_{1:n}) - \pi(\bb{X}_i))^j \kappa_{\bb{s},b}(\bb{Y}_i), \quad j\in \N.
\]
By the triangle inequality, the conditional independence of $\delta_i$ and $\bb{Y}_i$ given $\bb{X}_i$, and Cauchy-Schwarz, we have, for all $j\in \N$,
\[
\begin{aligned}
\EE[M_{n,j}^2]
&\leq \frac{n^{-2}}{\pi_{\mathrm{min}}^{2(j+1)}} \sum_{i=1}^n \EE\left[\EE[|\hat{\pi}_i(\bb{X}_{1:n}) - \pi(\bb{X}_i)|^{2j} \mid \bb{X}_i] \, \EE[\kappa_{\bb{s},b}(\bb{Y}_i)^2 \mid \bb{X}_i]\right] \\
&\quad+ \frac{n^{-2}}{\pi_{\mathrm{min}}^{2(j+1)}} \sum_{i=1}^n \sum_{\substack{k=1 \\ k\neq i}}^n \EE\Big[\sqrt{\EE[|\hat{\pi}_i(\bb{X}_{1:n}) - \pi(\bb{X}_i)|^{2j} \mid \bb{X}_i, \bb{X}_k]} \Big. \\[-3mm]
&\hspace{33mm}\times \Big. \sqrt{\EE[|\hat{\pi}_k(\bb{X}_{1:n}) - \pi(\bb{X}_k)|^{2j} \mid \bb{X}_i, \bb{X}_k]} \, \EE[\kappa_{\bb{s},b}(\bb{Y}_i) \mid \bb{X}_i] \, \EE[\kappa_{\bb{s},b}(\bb{Y}_k) \mid \bb{X}_k]\Big].
\end{aligned}
\]
Combined with \eqref{eq:Geenens} and the upper bound on the $L^2$ norm of the Dirichlet kernel in Lemma~\ref{lem:kappa.L2.norm} under Assumption~\ref{ass:6}, with $f_{\bb{Y}_1 \mid \bb{X}_1}$ as the target density, we have
\begin{equation}\label{eq:bound.M2}
\EE[M_{n,j}^2] \ll n^{-1} n^{-4j/(p+4)} b^{-d/2} + n^{-4j/(p+4)} \ll n^{-4j/(p+4)},
\end{equation}
where the last bound is due to $n b^{d/2}\to \infty$ in Assumption~\ref{ass:1}. By the Cauchy-Schwarz inequality, it follows that, for all $j,j'\in \N$,
\[
|\Cov(M_{n,j},M_{n,j'})|
\leq \sqrt{\Var(M_{n,j})} \sqrt{\Var(M_{n,j'})}
\leq \sqrt{\EE[M_{n,j}^2]} \sqrt{\EE[M_{n,j'}^2]} \ll n^{-2(j+j')/(p+4)}.
\]
A similar argument yields, for all $j\in \N$,
\[
|\Cov(\widetilde{f}_{n,b}(\bb{s}),M_{n,j})| \ll n^{-2j/(p+4)}.
\]
Thus, putting the last three equations back into \eqref{eq:hat.f.expansion.2}, we get
\begin{equation}\label{eq:star}
\Var(\hat{f}_{n,b}(\bb{s}))
= \Var(\widetilde{f}_{n,b}(\bb{s})) - 2 \Cov(\widetilde{f}_{n,b}(\bb{s}), M_{n,1}) + \Var(M_{n,1}) + \OO(n^{-4/(p+4)}).
\end{equation}

Next, our goal is to decompose $M_{n,1}$. Letting $\hat{g}_n(\bb{x})$ be defined as in \eqref{eq:g.hat}, we have
\begin{equation}\label{eq:M1.U1.U2}
\begin{aligned}
M_{n,1}
&= n^{-1} \sum_{i=1}^n \frac{\delta_i}{\pi(\bb{X}_i)^2} \big(\hat{\pi}_i(\bb{X}_{1:n}) - \pi(\bb{X}_i)\big) \kappa_{\bb{s},b}(\bb{Y}_i) \\
&= n^{-2} \sum_{i=1}^n \sum_{j=1}^n \frac{\delta_i (\delta_j - \pi(\bb{X}_i))}{\pi(\bb{X}_i)^2} \frac{K_h^*(\bb{X}_i - \bb{X}_j)}{\hat{g}_n(\bb{X}_i)} \kappa_{\bb{s},b}(\bb{Y}_i) \\
&\equiv U_{n,1} + U_{n,2},
\end{aligned}
\end{equation}
where
\[
U_{n,1} = n^{-2} \sum_{i=1}^n \sum_{j=1}^n \frac{(\delta_i - \pi(\bb{X}_i))(\delta_j - \pi(\bb{X}_i))}{\pi(\bb{X}_i)^2} \frac{K_h^*(\bb{X}_i - \bb{X}_j)}{\hat{g}_n(\bb{X}_i)} \kappa_{\bb{s},b}(\bb{Y}_i),
\]
and
\[
U_{n,2} = n^{-2} \sum_{i=1}^n \sum_{j=1}^n \frac{(\delta_j - \pi(\bb{X}_i))}{\pi(\bb{X}_i)} \frac{K_h^*(\bb{X}_i - \bb{X}_j)}{\hat{g}_n(\bb{X}_i)} \kappa_{\bb{s},b}(\bb{Y}_i).
\]

First, we bound $\EE[U_{n,2}^2]$. Consider the recentering,
\begin{equation}\label{eq:recentering}
\delta_j - \pi(\bb{X}_i) = (\delta_j - \pi(\bb{X}_j)) + (\pi(\bb{X}_j) - \pi(\bb{X}_i)),
\end{equation}
which gives us the following decomposition:
\begin{equation}\label{eq:U2.V1.V2}
\begin{aligned}
U_{n,2}
&= n^{-2} \sum_{i=1}^n \sum_{j=1}^n \frac{(\delta_j - \pi(\bb{X}_j))}{\pi(\bb{X}_i)} \frac{K_h^*(\bb{X}_i - \bb{X}_j)}{\hat{g}_n(\bb{X}_i)} \kappa_{\bb{s},b}(\bb{Y}_i) \\
&\quad+ n^{-2} \sum_{i=1}^n \sum_{j=1}^n \frac{(\pi(\bb{X}_j) - \pi(\bb{X}_i))}{\pi(\bb{X}_i)} \frac{K_h^*(\bb{X}_i - \bb{X}_j)}{\hat{g}_n(\bb{X}_i)} \kappa_{\bb{s},b}(\bb{Y}_i) \\
&\equiv V_{n,1} + V_{n,2}.
\end{aligned}
\end{equation}

Now, we want to bound $\EE[V_{n,2}^2]$. Rewrite $V_{n,2}$ as
\[
V_{n,2} = n^{-1} \sum_{i=1}^n \frac{\kappa_{\bb{s},b}(\bb{Y}_i)}{\pi(\bb{X}_i)} D_{n,i},
\]
with
\[
D_{n,i} = \frac{n^{-1} \sum_{j=1}^n (\pi(\bb{X}_j) - \pi(\bb{X}_i)) K_h^*(\bb{X}_i - \bb{X}_j)}{\hat{g}_n(\bb{X}_i)}.
\]
Under Assumptions~\ref{ass:B.1}--\ref{ass:B.3}, for any point $\bb{x}\in \R^p$ such that $g(\bb{x})\in (0,\infty)$, a second-order Taylor expansion yields
\begin{equation}\label{eq:Vn2.mean.kernel}
\left|\EE\left[(\pi(\bb{X}_1) - \pi(\bb{x})) K_h^*(\bb{x} - \bb{X}_1)\right]\right| \ll h^2 .
\end{equation}
(The term of order $h$ is exactly $0$ by the symmetry of $K^*$.) Moreover, since $0 \leq \pi(\cdot) \leq 1$, $K^*$ is bounded by Assumption~\ref{ass:B.3}, and $g$ has bounded support by Assumption~\ref{ass:4} and is bounded by Assumption~\ref{ass:B.2}, we also have
\[
\begin{aligned}
\EE\left[(\pi(\bb{X}_1) - \pi(\bb{x}))^2 K_h^*(\bb{x} - \bb{X}_1)^2\right]
\leq \EE\left[K_h^*(\bb{x} - \bb{X}_1)^2\right]
&= \int_{\R^p} h^{-2p} K^*(h^{-1} (\bb{x} - \bb{t}))^2 g(\bb{t}) \rd \bb{t} \\
&= h^{-p} \int_{\R^p} K^*(\bb{z})^2 g(\bb{x} - h \bb{z}) \rd \bb{z} \\
&\ll h^{-p}.
\end{aligned}
\]
It follows from the last two equations, together with the independence of the $\bb{X}_i$'s, that, for all $i\in [n]$,
\[
\begin{aligned}
&\EE\left[\left\{n^{-1} \sum_{j=1}^n (\pi(\bb{X}_j) - \pi(\bb{X}_i)) K_h^*(\bb{X}_i - \bb{X}_j)\right\}^2 \mid \bb{X}_i\right] \\
&\quad\ll n^{-2} \sum_{j=1}^n \EE\left[(\pi(\bb{X}_j) - \pi(\bb{X}_i))^2 K_h^*(\bb{X}_i - \bb{X}_j)^2 \mid \bb{X}_i\right] \\
&\qquad+ \left\{n^{-1} \sum_{j=1}^n \left|\EE\left[(\pi(\bb{X}_j) - \pi(\bb{X}_i)) K_h^*(\bb{X}_i - \bb{X}_j) \mid \bb{X}_i\right]\right|\right\}^2 \\
&\quad\ll n^{-1} h^{-p} + h^4 \\
&\quad\ll n^{-4/(p+4)},
\end{aligned}
\]
where the last line is a consequence of our choice $h\sim \kappa n^{-1/(p+4)}$. Since $g$ is bounded below by $g_{\min}$ by Assumption~\ref{ass:4}, we deduce
\[
\EE[D_{n,i}^2] \ll \frac{1}{g_{\mathrm{min}}^2} \EE\left[\left\{n^{-1} \sum_{j=1}^n (\pi(\bb{X}_j) - \pi(\bb{X}_i)) K_h^*(\bb{X}_i - \bb{X}_j)\right\}^2\right] \ll n^{-4/(p+4)}.
\]
Finally, conditioning on $\bb{X}_{1:n}$ and using Assumption~\ref{ass:5}, we get
\[
\begin{aligned}
\EE[V_{n,2}^2]
&\ll n^{-2} \sum_{i=1}^n \EE\left[\EE[\kappa_{\bb{s},b}(\bb{Y}_i)^2 \mid \bb{X}_i] \, D_{n,i}^2\right] \\
&\quad+ n^{-2} \sum_{i=1}^n \sum_{\substack{j=1 \\ j\neq i}}^n \EE\big[\EE[\kappa_{\bb{s},b}(\bb{Y}_i) \mid \bb{X}_i] \EE[\kappa_{\bb{s},b}(\bb{Y}_j) \mid \bb{X}_j] \, |D_{n,i}| |D_{n,j}|\big].
\end{aligned}
\]
Now, using $\EE[\kappa_{\bb{s},b}(\bb{Y}_1)^2 \mid \bb{X}_1] \ll b^{-d/2}$ from Lemma~\ref{lem:kappa.L2.norm} (under Assumption~\ref{ass:6}) and the boundedness of $f_{\bb{Y}_1 \mid \bb{X}_1}$ in Assumption~\ref{ass:6}, followed by the Cauchy-Schwarz inequality, we have
\begin{equation}\label{eq:Vn2.bound}
\begin{aligned}
\EE[V_{n,2}^2]
&\ll n^{-2} b^{-d/2} \sum_{i=1}^n \EE[D_{n,i}^2] + n^{-2} \sum_{\substack{i,j=1 \\ j\neq i}}^n \sqrt{\EE[D_{n,i}^2]} \sqrt{\EE[D_{n,j}^2]} \\
&\ll n^{-1} b^{-d/2} n^{-4/(p+4)} + n^{-4/(p+4)} \\
&\ll n^{-4/(p+4)}.
\end{aligned}
\end{equation}

For $V_{n,1}$ in \eqref{eq:U2.V1.V2}, define, for every $\bb{x}\in \R^p$ such that $g(\bb{x})\in (0,\infty)$,
\[
\hat{q}_n(\bb{x}) = \frac{\sum_{i=1}^n \kappa_{\bb{s},b}(\bb{Y}_i) K_h^*(\bb{x} - \bb{X}_i)}{\sum_{i=1}^n K_h^*(\bb{x} - \bb{X}_i)},
\]
which is the Nadaraya--Watson estimator of the regression function $q(\bb{x}) = \EE[\kappa_{\bb{s},b}(\bb{Y}_1) \mid \bb{X}_1 = \bb{x}]$. We can rewrite $V_{n,1}$ as
\begin{equation}\label{eq:V1.R1.R2}
\begin{aligned}
V_{n,1}
&= n^{-1} \sum_{j=1}^n (\delta_j - \pi(\bb{X}_j)) \, n^{-1} \sum_{i=1}^n \frac{\kappa_{\bb{s},b}(\bb{Y}_i) K_h^*(\bb{X}_j - \bb{X}_i)}{\pi(\bb{X}_i) \hat{g}_n(\bb{X}_i)} \\
&= n^{-1} \sum_{j=1}^n \frac{\delta_j - \pi(\bb{X}_j)}{\pi(\bb{X}_j)} \hat{q}_n(\bb{X}_j) + R_{n,1} \\
&= n^{-1} \sum_{j=1}^n \frac{\delta_j - \pi(\bb{X}_j)}{\pi(\bb{X}_j)} q(\bb{X}_j) + R_{n,1} + R_{n,2},
\end{aligned}
\end{equation}
where
\[
R_{n,1} = n^{-1} \sum_{j=1}^n (\delta_j - \pi(\bb{X}_j)) S_{n,j}, \qquad R_{n,2} = n^{-1} \sum_{j=1}^n \frac{\delta_j - \pi(\bb{X}_j)}{\pi(\bb{X}_j)} \big\{\hat{q}_n(\bb{X}_j) - q(\bb{X}_j)\big\},
\]
with
\[
S_{n,j} = n^{-1} \sum_{i=1}^n K_h^*(\bb{X}_j - \bb{X}_i) \kappa_{\bb{s},b}(\bb{Y}_i) \left\{\frac{1}{\pi(\bb{X}_i)\hat{g}_n(\bb{X}_i)} - \frac{1}{\pi(\bb{X}_j)\hat{g}_n(\bb{X}_j)}\right\}.
\]

We now bound $\EE[R_{n,1}^2]$. By Assumptions~\ref{ass:4}--\ref{ass:5} and \ref{ass:B.2}, we have the first-order Taylor expansion
\[
\frac{1}{\pi(\bb{X}_i) g(\bb{X}_i)} - \frac{1}{\pi(\bb{X}_j) g(\bb{X}_j)}
= \nabla\left(\frac{1}{\pi g}\right)(\bb{X}_j)^{\top}(\bb{X}_i - \bb{X}_j) + \OO(\|\bb{X}_i - \bb{X}_j\|_2^2),
\]
with $\|\nabla(1/(\pi g))(\bb{x})\|_{\infty}$ uniformly bounded in $\bb{x}$. Similarly to \eqref{eq:Vn2.mean.kernel}, using the boundedness of $f_{\bb{Y}_1 \mid \bb{X}_1}$ in Assumption~\ref{ass:6}, and the symmetry and moment conditions of $K^*$ in Assumption~\ref{ass:B.3} leads to
\[
\begin{aligned}
\EE\left[n^{-1} \sum_{i=1}^n K_h^*(\bb{X}_j - \bb{X}_i) \kappa_{\bb{s},b}(\bb{Y}_i) \left\{\frac{1}{\pi(\bb{X}_i) g(\bb{X}_i)} - \frac{1}{\pi(\bb{X}_j) g(\bb{X}_j)}\right\} \mid \bb{X}_{1:n}\right]
&\ll h^2 \, \EE[\kappa_{\bb{s},b}(\bb{Y}_1) \mid \bb{X}_1] \\
&\ll h^2.
\end{aligned}
\]
Conditionally on $\bb{X}_{1:n}$, the summands are independent, so the corresponding variance satisfies
\[
\Var\left(n^{-1} \sum_{i=1}^n K_h^*(\bb{X}_j - \bb{X}_i) \kappa_{\bb{s},b}(\bb{Y}_i) \left\{\frac{1}{\pi(\bb{X}_i) g(\bb{X}_i)} - \frac{1}{\pi(\bb{X}_j) g(\bb{X}_j)}\right\} \mid \bb{X}_{1:n}\right) \ll n^{-1} h^{-p} b^{-d/2}.
\]
By Assumptions~\ref{ass:B.1}--\ref{ass:B.3}, the uniform strong consistency of $\hat{g}_n$, and Assumption~\ref{ass:4}, it follows from the last two equations that, for all $j\in [n]$,
\[
\EE\left[S_{n,j}^2 \mid \bb{X}_{1:n}\right]
\ll n^{-1} h^{-p} b^{-d/2} + h^4.
\]
Since
\[
\EE[(\delta_j - \pi(\bb{X}_j))^2 \mid \bb{X}_{1:n}, \bb{Y}_{1:n}] \leq 1, \quad
\EE[(\delta_j - \pi(\bb{X}_j)) (\delta_k - \pi(\bb{X}_k)) \mid \bb{X}_{1:n}, \bb{Y}_{1:n}] = 0,
\]
for all $j\neq k$, we have
\begin{equation}\label{eq:R2.as.R1}
\begin{aligned}
\EE[R_{n,1}^2 \mid \bb{X}_{1:n}, \bb{Y}_{1:n}]
&= n^{-2} \sum_{j=1}^n \EE[(\delta_j - \pi(\bb{X}_j))^2 \mid \bb{X}_{1:n}, \bb{Y}_{1:n}] S_{n,j}^2 \\
&\quad+ n^{-2} \sum_{j=1}^n \sum_{\substack{k=1 \\ k\neq j}}^n \EE[(\delta_j - \pi(\bb{X}_j)) (\delta_k - \pi(\bb{X}_k)) \mid \bb{X}_{1:n}, \bb{Y}_{1:n}] S_{n,j} S_{n,k} \\
&\leq n^{-2} \sum_{j=1}^n S_{n,j}^2,
\end{aligned}
\end{equation}
so that
\begin{equation}\label{eq:Rn1.bound}
\begin{aligned}
\EE[R_{n,1}^2]
\leq \EE\left[n^{-2} \sum_{j=1}^n \EE\left[S_{n,j}^2 \mid \bb{X}_{1:n}\right]\right]
\ll n^{-1} (n^{-1} h^{-p} b^{-d/2} + h^4)
\ll n^{-4/(p+4)} n^{-4/(d+4)},
\end{aligned}
\end{equation}

We now bound $\EE[R_{n,2}^2]$. Proceeding as in \eqref{eq:R2.as.R1}, we have
\begin{equation}\label{eq:L2.q.hat.q}
\EE\left[R_{n,2}^2 \mid \bb{X}_{1:n}, \bb{Y}_{1:n}\right]
\leq n^{-2} \sum_{j=1}^n \frac{1}{\pi(\bb{X}_j)^2} \big\{\hat{q}_n(\bb{X}_j) - q(\bb{X}_j)\big\}^2.
\end{equation}
Under Assumptions~\ref{ass:B.1}--\ref{ass:B.3} and \ref{ass:C.1}, we have, using the local bound on the Dirichlet kernel in Lemma~\ref{lem:local.bound} and standard arguments for Nadaraya--Watson regression,
\begin{equation}\label{eq:qhat.mse}
\EE\Big[\big\{\hat{q}_n(\bb{x}) - q(\bb{x})\big\}^2\Big] \ll n^{-1} h^{-p} b^{-d/2} + h^4 \ll h^{-p} n^{-4/(d+4)} + h^4.
\end{equation}
Taking expectations in \eqref{eq:L2.q.hat.q}, and using \eqref{eq:qhat.mse} in conjunction with Assumption~\ref{ass:5} then yields
\begin{equation}\label{eq:Rn2.bound}
\EE[R_{n,2}^2]
\ll n^{-1} (h^{-p} n^{-4/(d+4)} + h^4)
\ll n^{-4/(p+4)} n^{-4/(d+4)}.
\end{equation}

Now, we go back to \eqref{eq:M1.U1.U2} to bound $\EE[U_{n,1}^2]$. The same recentering argument \eqref{eq:recentering} that we used for $U_{n,2}$ gives us
\begin{equation}\label{eq:Un1.decomp}
\begin{aligned}
U_{n,1}
&= n^{-2} \sum_{i=1}^n \sum_{j=1}^n \frac{(\delta_i - \pi(\bb{X}_i))(\delta_j - \pi(\bb{X}_j))}{\pi(\bb{X}_i)^2} \frac{K_h^*(\bb{X}_i - \bb{X}_j)}{\hat{g}_n(\bb{X}_i)} \kappa_{\bb{s},b}(\bb{Y}_i) \\
&\quad+ n^{-2} \sum_{i=1}^n \sum_{j=1}^n \frac{(\delta_i - \pi(\bb{X}_i))(\pi(\bb{X}_j) - \pi(\bb{X}_i))}{\pi(\bb{X}_i)^2} \frac{K_h^*(\bb{X}_i - \bb{X}_j)}{\hat{g}_n(\bb{X}_i)} \kappa_{\bb{s},b}(\bb{Y}_i) \\
&\equiv W_{n,1} + W_{n,2},
\end{aligned}
\end{equation}
with
\begin{equation}\label{eq:Wn2.bound}
\EE[W_{n,2}^2] \ll n^{-4/(p+4)}.
\end{equation}
To control $W_{n,1}$, note that as in \eqref{eq:cov.delta}, for indices $i,j,k,\ell\in \{1,\ldots,n\}$ that are all different,
\[
\begin{aligned}
\EE[(\delta_i - \pi(\bb{X}_i))(\delta_j - \pi(\bb{X}_j))(\delta_k - \pi(\bb{X}_k))(\delta_{\ell} - \pi(\bb{X}_{\ell})) \mid \bb{X}_{1:n}] &= 0, \\
\EE[(\delta_i - \pi(\bb{X}_i))^2(\delta_j - \pi(\bb{X}_j))(\delta_k - \pi(\bb{X}_k)) \mid \bb{X}_{1:n}] &= 0, \\
\EE[(\delta_i - \pi(\bb{X}_i))^3(\delta_j - \pi(\bb{X}_j)) \mid \bb{X}_{1:n}] &= 0.
\end{aligned}
\]
Therefore, when expanding the square of the double sum inside $\EE[W_{n,1}^2]$ into a quadruple sum, say $\sum_{i,j,k,\ell=1}^n$, the only $4$-tuples $(i,j,k,\ell)$ that survive, after conditioning on $\bb{X}_{1:n}$ and evaluating the conditional expectations, are the ones where all the indices are equal ($i=j=k=\ell$), or where there are two pairs of equal indices. Using Assumptions~\ref{ass:B.1}--\ref{ass:B.3}, the upper bound on the $L^2$ norm of the Dirichlet kernel in Lemma~\ref{lem:kappa.L2.norm} (under Assumption~\ref{ass:6}), and the boundedness of $f_{\bb{Y}_1 \mid \bb{X}_1}$ in Assumption~\ref{ass:6}, we obtain
\begin{equation}\label{eq:Wn1.bound}
\begin{aligned}
\EE[W_{n,1}^2]
&\ll \frac{n^{-4}}{\pi_{\mathrm{min}}^4 g_{\mathrm{min}}^2} \left\{\sum_{\substack{i,j,k,\ell=1 \\ i=j=k=\ell}}^n \EE\big[K_h^*(\bb{0}_p)^2 \EE[\kappa_{\bb{s},b}(\bb{Y}_i)^2 \mid \bb{X}_i]\big] \right. \\
&\hspace{23mm}\left.+ \sum_{\substack{i,j,k,\ell=1 \\ i\neq j,i=k,j=\ell}}^n \EE\big[K_h^*(\bb{X}_i - \bb{X}_j)^2 \EE[\kappa_{\bb{s},b}(\bb{Y}_i)^2 \mid \bb{X}_i]\big] \right. \\
&\hspace{23mm}\left.+ \sum_{\substack{i,j,k,\ell=1 \\ i\neq k,i=j,k=\ell}}^n \EE\big[K_h^*(\bb{0}_p)^2 \EE[\kappa_{\bb{s},b}(\bb{Y}_i) \mid \bb{X}_i] \EE[\kappa_{\bb{s},b}(\bb{Y}_k) \mid \bb{X}_k]\big] \right. \\[-3mm]
&\hspace{23mm}\left.+ \sum_{\substack{i,j,k,\ell=1 \\ i\neq k,i=\ell,j=k}}^n \EE\big[K_h^*(\bb{X}_i - \bb{X}_j)^2 \EE[\kappa_{\bb{s},b}(\bb{Y}_i) \mid \bb{X}_i] \EE[\kappa_{\bb{s},b}(\bb{Y}_k) \mid \bb{X}_k]\big]\right\} \\
&\ll \frac{n^{-4}}{\pi_{\mathrm{min}}^4 g_{\mathrm{min}}^2} \Big\{n h^{-2p} \, K^*(\bb{0}_p)^2 \, b^{-d/2} + n^2 \, \EE\big[K_h^*(\bb{X}_1 - \bb{X}_2)^2\big] \, b^{-d/2} \Big. \\[-1mm]
&\hspace{58mm}\Big.+ n^2 h^{-2p} \, K^*(\bb{0}_p)^2 + n^2 \, \EE\big[K_h^*(\bb{X}_1 - \bb{X}_2)^2\big]\Big\} \\
&\ll n^{-4} (n h^{-2p} b^{-d/2} + n^2 h^{-p} b^{-d/2} + n^2 h^{-2p} + n^2 h^{-p}) \\
&\ll n^{-4/(p+4)} n^{-4/(d+4)} + n^{-8/(p+4)},
\end{aligned}
\end{equation}
where the last bound is due to our choices $b\sim c n^{-2/(d+4)}$ and $h\sim \kappa n^{-1/(p+4)}$. Putting \eqref{eq:Wn1.bound} and \eqref{eq:Wn2.bound} back into \eqref{eq:Un1.decomp} yields
\begin{equation}\label{eq:Un1.bound}
\begin{aligned}
\EE[U_{n,1}^2]
&\ll \EE[W_{n,1}^2] + \EE[W_{n,2}^2] \\
&\ll (n^{-4/(p+4)} n^{-4/(d+4)} + n^{-8/(p+4)}) + n^{-4/(p+4)} \\
&\ll n^{-4/(p+4)}.
\end{aligned}
\end{equation}

Combining \eqref{eq:M1.U1.U2}, \eqref{eq:U2.V1.V2}, and \eqref{eq:V1.R1.R2} gives us the representation:
\[
M_{n,1}
= n^{-1} \sum_{i=1}^n \frac{(\delta_i - \pi(\bb{X}_i))}{\pi(\bb{X}_i)} \EE[\kappa_{\bb{s},b}(\bb{Y}_i) \mid \bb{X}_i] + R_n, \qquad
R_n = U_{n,1} + V_{n,2} + R_{n,1} + R_{n,2}.
\]
By \eqref{eq:Un1.bound}, \eqref{eq:Vn2.bound}, \eqref{eq:Rn1.bound}, and \eqref{eq:Rn2.bound}, we have
\begin{equation}\label{eq:Rn.bound.feasible}
\begin{aligned}
\EE[R_n^2]
&\ll \EE[U_{n,1}^2] + \EE[V_{n,2}^2] + \EE[R_{n,1}^2] + \EE[R_{n,2}^2] \\
&\ll n^{-4/(p+4)} + n^{-4/(p+4)} + n^{-4/(p+4)} n^{-4/(d+4)} + n^{-4/(p+4)} n^{-4/(d+4)} \\
&\ll n^{-4/(p+4)}.
\end{aligned}
\end{equation}

To conclude the proof of the proposition, write
\begin{equation}\label{eq:M.B.R}
M_{n,1} = B_n + R_n, \qquad
B_n = n^{-1} \sum_{i=1}^n \frac{\delta_i - \pi(\bb{X}_i)}{\pi(\bb{X}_i)} \EE[\kappa_{\bb{s},b}(\bb{Y}_i) \mid \bb{X}_i].
\end{equation}
It remains to show that we have, up to some error terms to be determined explicitly below,
\begin{align}
\Var(M_{n,1}) &\approx \Var(B_n), \label{eq:to.prove.1} \\
\Cov(\widetilde{f}_{n,b}(\bb{s}), M_{n,1}) &\approx \Var(B_n), \label{eq:to.prove.2}
\end{align}
so that \eqref{eq:star} becomes
\begin{equation}
\Var(\hat{f}_{n,b}(\bb{s})) \approx \Var(\widetilde{f}_{n,b}(\bb{s})) - \Var(B_n),
\end{equation}
which is the claim of the theorem.

To this end, recall that $\EE[\delta_i - \pi(\bb{X}_i) \mid \bb{X}_i] = 0$ and $\Var(\delta_i \mid \bb{X}_i) = \pi(\bb{X}_i) (1 - \pi(\bb{X}_i))$, so we have $\EE[B_n] = 0$, and by the independence of the pairs $(\bb{X}_i,\bb{Y}_i)$ and the law of total variance,
\begin{equation}\label{eq:Var.Bn.feasible}
\begin{aligned}
\Var(B_n)
&= n^{-2} \sum_{i=1}^n \left\{\EE\left[\Var\left(\frac{\delta_i - \pi(\bb{X}_i)}{\pi(\bb{X}_i)} \EE[\kappa_{\bb{s},b}(\bb{Y}_i) \mid \bb{X}_i] \mid \bb{X}_i\right)\right] + 0\right\} \\
&= n^{-2} \sum_{i=1}^n \EE\left[\frac{\Var(\delta_i \mid \bb{X}_i)}{\pi(\bb{X}_i)^2} \{\EE[\kappa_{\bb{s},b}(\bb{Y}_i) \mid \bb{X}_i]\}^2\right] \\
&= n^{-2} \sum_{i=1}^n \EE\left[\frac{1 - \pi(\bb{X}_i)}{\pi(\bb{X}_i)} \{\EE[\kappa_{\bb{s},b}(\bb{Y}_i) \mid \bb{X}_i]\}^2\right].
\end{aligned}
\end{equation}
By the representation \eqref{eq:M.B.R} and Cauchy--Schwarz, note that
\begin{equation}\label{eq:CS}
\begin{aligned}
\left|\Var(M_{n,1}) - \Var(B_n)\right|
&\leq 2 \, |\Cov(B_n,R_n)| + \Var(R_n) \\
&\leq 2 \sqrt{\Var(B_n)} \sqrt{\EE[R_n^2]} + \EE[R_n^2].
\end{aligned}
\end{equation}
By Assumptions~\ref{ass:5} and \ref{ass:6}, we have $\Var(B_n) \ll n^{-1}$. Combined with \eqref{eq:CS} and \eqref{eq:Rn.bound.feasible}, this yields
\begin{equation}\label{eq:end.0}
\Var(M_{n,1}) = \Var(B_n) + \OO(n^{-1/2} n^{-2/(p+4)}) + \OO(n^{-4/(p+4)}),
\end{equation}
which proves \eqref{eq:to.prove.1} with explicit errors.

Next, using $M_{n,1} = B_n + R_n$ again, we have
\begin{equation}\label{eq:end.1}
\Cov(\widetilde{f}_{n,b}(\bb{s}), M_{n,1})
= \Cov(\widetilde{f}_{n,b}(\bb{s}), B_n)
+ \OO\left(\sqrt{\Var(\widetilde{f}_{n,b}(\bb{s}))} \sqrt{\EE[R_n^2]}\right).
\end{equation}
By Proposition~\ref{prop:variance.pseudo} and \eqref{eq:Rn.bound.feasible}, the error term is bounded by
\begin{equation}\label{eq:end.2}
\sqrt{\Var(\widetilde{f}_{n,b}(\bb{s}))} \sqrt{\EE[R_n^2]}
\ll_{\bb{s}} \sqrt{n^{-1} b^{-d/2}} \sqrt{n^{-4/(p+4)}}
\ll n^{-2/(d+4)} n^{-2/(p+4)}.
\end{equation}
Now, since $\EE[\delta_i - \pi(\bb{X}_i) \mid \bb{X}_i] = 0$ and $\delta_i$ is conditionally independent of $\bb{Y}_i$ given $\bb{X}_i$, we have
\[
\begin{aligned}
\Cov(\widetilde{f}_{n,b}(\bb{s}), B_n)
&= n^{-2} \sum_{i=1}^n \Cov\left(\frac{\delta_i}{\pi(\bb{X}_i)} \kappa_{\bb{s},b}(\bb{Y}_i), \frac{\delta_i - \pi(\bb{X}_i)}{\pi(\bb{X}_i)} \EE[\kappa_{\bb{s},b}(\bb{Y}_i) \mid \bb{X}_i]\right) \\
&= n^{-2} \sum_{i=1}^n \EE\left[\frac{\delta_i (\delta_i - \pi(\bb{X}_i))}{\pi(\bb{X}_i)^2} \kappa_{\bb{s},b}(\bb{Y}_i) \EE[\kappa_{\bb{s},b}(\bb{Y}_i) \mid \bb{X}_i]\right] \\
&= n^{-2} \sum_{i=1}^n \EE\left[\EE\left[\frac{\delta_i (\delta_i - \pi(\bb{X}_i))}{\pi(\bb{X}_i)^2} \mid \bb{X}_i\right] \{\EE[\kappa_{\bb{s},b}(\bb{Y}_i) \mid \bb{X}_i]\}^2\right].
\end{aligned}
\]
Furthermore, $\delta_i^2 = \delta_i$, so we obtain
\[
\EE\left[\frac{\delta_i (\delta_i - \pi(\bb{X}_i))}{\pi(\bb{X}_i)^2} \mid \bb{X}_i\right]
= \EE\left[\frac{\delta_i}{\pi(\bb{X}_i)^2} - \frac{\delta_i}{\pi(\bb{X}_i)} \mid \bb{X}_i\right]
= \frac{1}{\pi(\bb{X}_i)} - 1
= \frac{1 - \pi(\bb{X}_i)}{\pi(\bb{X}_i)}.
\]
Therefore, by \eqref{eq:Var.Bn.feasible},
\begin{equation}\label{eq:end.3}
\Cov(\widetilde{f}_{n,b}(\bb{s}), B_n)
= n^{-2} \sum_{i=1}^n \EE\left[\frac{1 - \pi(\bb{X}_i)}{\pi(\bb{X}_i)} \{\EE[\kappa_{\bb{s},b}(\bb{Y}_i) \mid \bb{X}_i]\}^2\right]
= \Var(B_n).
\end{equation}
Putting \eqref{eq:end.1} and \eqref{eq:end.2} into \eqref{eq:end.3} yields
\begin{equation}\label{eq:end.4}
\Cov(\widetilde{f}_{n,b}(\bb{s}), M_{n,1}) = \Var(B_n) + \OO(n^{-2/(d+4)} n^{-2/(p+4)}),
\end{equation}
which proves \eqref{eq:to.prove.2} with explicit errors.

Finally, putting \eqref{eq:end.0} and \eqref{eq:end.4} back into \eqref{eq:star} yields
\[
\begin{aligned}
\Var(\hat{f}_{n,b}(\bb{s}))
&= \Var(\widetilde{f}_{n,b}(\bb{s})) - \Var(B_n) + \OO(n^{-1/2} n^{-2/(p+4)}) + \OO(n^{-4/(p+4)}) + \OO(n^{-2/(d+4)} n^{-2/(p+4)}) \\
&= \Var(\widetilde{f}_{n,b}(\bb{s})) - \Var(B_n) + \OO(n^{-4/(p+4)}) + \OO(n^{-2/(p+4)} n^{-2/(d+4)}).
\end{aligned}
\]
This completes the proof of the proposition.

\subsection{Proof of Theorem~\ref{thm:CLT.feasible}}

Let $\bb{s}\in \mathrm{Int}(\mathcal{S}_d)$ be given such that $f(\bb{s})\in (0,\infty)$. We choose $b\sim c n^{-2/(d+4)}$ and $h\sim \kappa n^{-1/(p+4)}$ for some $c,\kappa\in (0,\infty)$. Start from the decomposition
\[
\begin{aligned}
n^{1/2} b^{d/4} \frac{\hat{f}_{n,b}(\bb{s})-f(\bb{s})-b\phi(\bb{s})}{\sqrt{\psi(\bb{s}) f(\bb{s}) (1 + \zeta(\bb{s}))}}
&= n^{1/2} b^{d/4} \frac{\widetilde f_{n,b}(\bb{s})-f(\bb{s})-b\phi(\bb{s})}{\sqrt{\psi(\bb{s}) f(\bb{s}) (1 + \zeta(\bb{s}))}} \\
&\quad+ n^{1/2} b^{d/4} \frac{\hat{f}_{n,b}(\bb{s})-\widetilde f_{n,b}(\bb{s})}{\sqrt{\psi(\bb{s}) f(\bb{s}) (1 + \zeta(\bb{s}))}}.
\end{aligned}
\]
The first term on the right-hand side converges in law to $\mathcal{N}(0,1)$ by Theorem~\ref{thm:CLT.pseudo} with our choice of smoothing parameter $b\sim c n^{-2/(d+4)}$. To conclude, it is sufficient to prove the $L^2$ convergence:
\[
n^{1/2} b^{d/4} (\hat{f}_{n,b}(\bb{s}) - \widetilde f_{n,b}(\bb{s})) \stackrel{L^2}{\longrightarrow} 0, \quad n\to \infty.
\]

To compare $\hat{f}_{n,b}$ and $\widetilde f_{n,b}$, recall the identity \eqref{eq:hat.f.expansion.2}:
\[
\hat{f}_{n,b}(\bb{s}) - \widetilde f_{n,b}(\bb{s}) = \sum_{j=1}^{\infty} (-1)^j \, M_{n,j}, \qquad M_{n,j} = n^{-1} \sum_{i=1}^n \frac{\delta_i}{\pi(\bb{X}_i)^{j+1}} \big(\hat\pi_i(\bb{X}_{1:n})-\pi(\bb{X}_i)\big)^j \, \kappa_{\bb{s},b}(\bb{Y}_i).
\]
By the summable bound $\smash{\EE[M_{n,j}^2] \ll n^{-4j/(p+4)}}$ obtained in \eqref{eq:bound.M2} in the proof of Proposition~\ref{prop:variance.feasible} (with our choices of $b\sim c n^{-2/(d+4)}$ and $h\sim \kappa n^{-1/(p+4)}$), and the Cauchy-Schwarz inequality, we have
\[
\sum_{j=1}^{\infty} \EE[M_{n,j}^2] \ll \sum_{j=1}^{\infty} n^{-4j/(p+4)} \ll n^{-4/(p+4)},
\]
and
\[
\sum_{j=1}^{\infty} \sum_{\substack{j'=1 \\ j'\neq j}}^{\infty} |\EE[M_{n,j} M_{n,j'}]| \leq \sum_{j=1}^{\infty} \sum_{\substack{j'=1 \\ j'\neq j}}^{\infty} \sqrt{\EE[M_{n,j}^2]} \sqrt{\EE[M_{n,j'}^2]} \ll \sum_{j=1}^{\infty} \sum_{\substack{j'=1 \\ j'\neq j}}^{\infty} n^{-2(j+j')/(p+4)} \ll n^{-6/(p+4)}.
\]
It follows that if $p < d$, then
\[
\EE\left[\left\{n^{1/2} b^{d/4} (\hat{f}_{n,b}(\bb{s}) - \widetilde f_{n,b}(\bb{s}))\right\}^2\right] \ll n^{4/(d+4)} n^{-4/(p+4)} \to 0, \quad n\to \infty,
\]
completing the proof.

\section{Technical lemmas}\label{sec:tech.lemmas}

The first lemma studies the asymptotics of the $L^2$ norm of the Dirichlet kernel.

\begin{lemma}\label{lem:kappa.L2.norm}
Let $\bb{s}\in \mathrm{Int}(\mathcal{S}_d)$ be given. We have, as $b\to 0$,
\[
\left\{\int_{\mathcal{S}_d} \kappa_{\bb{s},b}(\bb{x})^2 \rd \bb{x}\right\}^{1/2} = b^{-d/4} \psi^{1/2}(\bb{s}) \{1 + \OO_{\bb{s}}(b)\},
\]
where $\psi$ is defined in \eqref{eq:psi.zeta}.
\end{lemma}

\begin{proof}[Proof of Lemma~\ref{lem:kappa.L2.norm}]
See, e.g., Lemma~1 of \citet{DKO2026}.
\end{proof}

The second lemma gives a uniform upper bound on the Dirichlet kernel $\bb{x}\mapsto \kappa_{\bb{s},b}(\bb{x})$ in $\mathcal{S}_d$.

\begin{lemma}\label{lem:local.bound}
Let $\bb{s}\in \mathcal{S}_d$ be given. We have, as $b\to 0$,
\[
\max_{\bb{x}\in \mathcal{S}_d} \kappa_{\bb{s},b}(\bb{x}) \ll b^{-d/2} \psi(\bb{s}),
\]
where $\psi$ is defined in \eqref{eq:psi.zeta}.
\end{lemma}

\begin{proof}[Proof of Lemma~\ref{lem:local.bound}]
See Lemma~2 of \citet{MR4319409}.
\end{proof}

\appendix

\section{List of acronyms}\label{app:acronyms}

\begin{tabular}{llll}
&alr    &\hspace{20mm} &additive log-ratio \\
&BMI    &\hspace{20mm} &body mass index \\
&CBC    &\hspace{20mm} &complete blood count \\
&iid    &\hspace{20mm} &independent and identically distributed \\
&ilr    &\hspace{20mm} &isometric log-ratio \\
&IPW    &\hspace{20mm} &inverse probability weighted \\
&IQR    &\hspace{20mm} &interquartile range \\
&ISE    &\hspace{20mm} &integrated squared error \\
&KDE    &\hspace{20mm} &kernel density estimator \\
&LSCV   &\hspace{20mm} &least-squares cross-validation \\
&MAR    &\hspace{20mm} &missing at random \\
&MSE    &\hspace{20mm} &mean squared error \\
&NHANES &\hspace{20mm} &National Health and Nutrition Examination Survey \\
&SD     &\hspace{20mm} &standard deviation
\end{tabular}

\section{Reproducibility}\label{app:code}

The \textsf{R} code that generated the figures, the simulation study results, and the real-data application is available online in the GitHub repository of \citet{DaayebOuimet2026github}.

\section*{Funding}

Fr\'ed\'eric Ouimet is supported by the Natural Sciences and Engineering Research Council of Canada (NSERC) through Discovery Grant RGPIN-2026-04471 and Discovery Launch Supplement DGECR-2026-00449.



\section*{References}
\addcontentsline{toc}{chapter}{References}

\setlength{\bibsep}{0pt plus 0.3ex}

\bibliographystyle{plainnat}
\bibliography{bib}

\end{document}